\documentclass[submission,copyright,creativecommons]{eptcs}
\providecommand{\event}{LSFA 2026}

\usepackage{iftex}
\ifpdf
  \usepackage{underscore}         
  \usepackage[T1]{fontenc}        
\else
  \usepackage{breakurl}           
\fi

\usepackage{xcolor}
\definecolor{blue4}{rgb}{0,0,0.545}

\usepackage{float}
\usepackage{graphicx}
\usepackage{subcaption}
\usepackage{tikz}
\usepackage{tikz-cd}
\usepackage{circuitikz}
\usetikzlibrary{petri,positioning, calc,decorations.pathmorphing,shapes,cd,arrows}

\usepackage{caption}
\usepackage{multirow}

\usepackage[inline]{enumitem}

\usepackage{algorithm}
\usepackage{algpseudocode}
\algrenewcommand\algorithmicrequire{\textbf{Input:}}
\algrenewcommand\algorithmicensure{\textbf{Output:}}
\newcommand{\code}[1]{\texttt{#1}}
\usepackage{listings}
\definecolor{IdrisBlue}{RGB}{0,0,180}
\definecolor{IdrisGreen}{RGB}{0,150,0}
\definecolor{IdrisOrange}{RGB}{173,131,43}
\lstdefinelanguage{Idris}{
  sensitive=true,
  alsoletter={'},
  keywords=[1]{data,record,where,constructor,auto,Either,Maybe},
  keywords=[2]{Nat,String,Fin,Vect,List,So,Type,Injective,Elem,LTE,integerLessThanNat,maximum,union},
  keywords=[3]{Span,Coproduct,Pushout,computeCoproduct,uniqueFromCoproduct,computePushout,SEQ,P_ASYNC,BRA_CLOSED,BRA_OPEN},
  keywords=[4]{Computon,Morphism,Subset,Trivial,Primitive,Surjective,Monomorphism,InMarker,OutMarker,vecI,vecO,inports,outports,image,mkTrivial,mkPrimitive,mkComputonMorphism,mkComputonMonomorphism,mkInMarker,mkOutMarker},
  keywordstyle=[1]\color{IdrisBlue}\bfseries,
  keywordstyle=[2]\color{IdrisGreen},
  keywordstyle=[3]\color{IdrisOrange},
  keywordstyle=[4]\color{IdrisOrange},  
  basicstyle=\ttfamily,
  frame=single
}
\usepackage{url}
\usepackage{siunitx}
\usepackage{amssymb,mathrsfs,nccmath,bbm,amsthm}
\usepackage{stmaryrd}
\usepackage{centernot}
\usepackage{mathpartir}
\newtheorem{notation}{Notation}
\newtheorem{definition}{Definition}
\newtheorem{proposition}{Proposition}
\newtheorem{remark}{Remark}
\newtheorem{lemma}{Lemma}
\newtheorem{theorem}{Theorem}
\newtheorem{example}{Example}

\hypersetup{
  colorlinks=true,
  linkcolor=blue4,    
  citecolor=blue4,    
  urlcolor=blue4      
}

\newcommand{\arrowflow}{
\tikz \draw[-{to}] (-1pt,0) -- (1pt,0);
}

\newcommand{\flow}[4]{ 
	\draw[#3] (#1) to [#4] node [pos=0.5] {\arrowflow} (#2);
}
\newcommand{\flowdiag}[5]{ 
	\draw[#3] (#1) to [#4] node [#5] {\arrowflow} (#2);
}

\tikzset{fill fraction/.style n args={1}{path picture={
 \fill[black] (path picture bounding box.south west) rectangle
 ($(path picture bounding box.north west)!0.5!(path picture bounding box.north
 east)$);}}
}
\newcommand{\qmatch}[4]{ 
	\node[draw,fill=white,fill fraction={black}{0.5},inner sep=1.2pt,minimum size=3pt,label=left:{\scriptsize #4}] (#1) at (#2,#3) {};
}
\newcommand{\dmatch}[5]{ 
	\node[circle,draw,fill=white,fill fraction={black}{0.5},inner sep=1.2pt,minimum size=3pt,label=#5:{\scriptsize #4}] (#1) at (#2,#3) {};
}
\newcommand{\qinplain}[4]{ 
\node[draw=black,fill=white,inner sep=0pt,minimum size=3pt,label=\scriptsize #4] (#1) at (#2,#3) {};
}
\newcommand{\qin}[4]{ 
\flow{{#2,#3}}{$(#2,#3)+(0.5,0)$}{densely dashed}{};
\qinplain{#1}{#2}{#3}{#4};
}
\newcommand{\dinplain}[4]{ 
\node[circle,draw=black,fill=white,inner sep=0pt,minimum size=3pt,label=left:\scriptsize #4] (#1) at (#2,#3) {};
}
\newcommand{\din}[4]{ 
\flow{{#2,#3}}{$(#2,#3)+(0.5,0)$}{}{};
\dinplain{#1}{#2}{#3}{#4};
}
\newcommand{\qoutplain}[4]{ 
\node[fill=black,inner sep=0pt,minimum size=3pt,label=\scriptsize #4] (#1) at (#2+0.5,#3) {};
}
\newcommand{\qout}[4]{ 
\flow{{#2,#3}}{$(#2,#3)+(0.5,0)$}{densely dashed}{};
\qoutplain{#1}{#2}{#3}{#4};
}
\newcommand{\doutplain}[4]{ 
\node[circle,fill=black,inner sep=0pt,minimum size=3pt,label=right:\scriptsize #4] (#1) at (#2+0.5,#3) {};
}
\newcommand{\dout}[4]{ 
\flow{{#2,#3}}{$(#2,#3)+(0.5,0)$}{}{};
\doutplain{#1}{#2}{#3}{#4};
}

\newcommand{\primitive}[4]{ 
\draw[draw=black,fill=white] (#1,#2) rectangle ++(#3,#4);
}
\newcommand{\computonComposite}[4]{ 
\draw[draw=black,fill=black!2] (#1,#2) rectangle ++(#3,#4);
}

\title{Colimit-Based Composition of High-Level Computing Devices}
\author{Damian Arellanes
\institute{Lancaster University}
\email{damian.arellanes@lancaster.ac.uk}}
\def\titlerunning{Colimit-Based Composition of High-Level Computing Devices}
\def\authorrunning{D. Arellanes}

\begin{document}
\maketitle         
\vspace{-5pt}
\begin{abstract}
Models of High-level Computation (MHCs) provide effective means to describe complex real-world computing systems because they offer formal foundations for the specification of interacting computing devices, as opposed to describing individual ones, which has been the focus of classical models such as Turing machines or the lambda calculus itself. Despite numerous proposals over the past half century, there is still no canonical MHC akin to Turing machines for (compositionally) reasoning about computation in the large. One of the major drawbacks of state- and data-oriented MHCs is that they extensively neglect control flow, a well-known semantic property that defines computation order. Only control-oriented MHCs treat control explicitly at the expense of ignoring data flow or assuming that data follows control. Mixing data and control within the same framework leads to inefficient methods for formal analysis and verification. To address this, the computon model has recently emerged as a category-theoretic MHC that separates data and control and makes control explicit by supporting composition operators characterised as finite colimit constructions. Such constructions allow the formation of sequential, parallel, branching and iterative computing devices. Unfortunately, the computon model is still a generic reference rather than a concrete realisation. In this paper, we provide a variation of it to enable functional computing devices, introduce a new branching operator, discuss how to define synchronous parallelising out of sequencing and asynchronous parallelising, describe concrete operational semantics for computon execution and provide the first implementation of the model. The implementation yields an open-source programming environment that realises the underlying categorical semantics with partial type-level guarantees. This tool is publicly available for building complex computing devices with a high degree of structural correctness by construction.
\end{abstract}

\section{Introduction}
\label{sec:Introduction}

The Church-Turing thesis states that a function is computable if there is an effective procedure able to produce its values. Such a thesis has been successful for reasoning about computation in the small where a single abstract device realises a well-defined effective procedure. Over the past half century, there has been a collective attempt to move from the small to the large by defining models of computation capable of describing not the behaviour of isolated computing devices, but interactions among a collection of them. This paradigm shift, referred to as Models of High-Level Computation (MHCs) \cite{arellanesmodels2025}, has become increasingly relevant due to the need of describing complex computations in real-world domains (e.g., distributed systems). To date, there is no standard MHC akin to Turing machines for constructing or reasoning about computation in the large, despite the wide variety of MHCs that have emerged over time, e.g., Kahn Process Networks \cite{kahnsemantics1974}, Hierarchical State Machines (HSMs) \cite{harelstatecharts1987}, ONets \cite{baezopen2020} and, more recently, algebras over operads of wiring diagrams \cite{yauoperads2018}.

In the context of MHCs, interaction is the causal effect of composition, an inductive mechanism for gluing together computing devices into more complex ones known as composites. For example, HSMs compose state machines by nesting, ONets compose Petri nets by identifying outputs with inputs via pushout constructions and Kahn networks combine processes by connecting them through (possibly unbounded) FIFO channels. In general, composition can be realised through the combination of states, control flow, data flow or any combination thereof. Although control lies at the heart of computation, because it is ever present in any (low- or even high-level) computing device, it is striking that such a dimension has been treated implicitly by some classes of MHCs \cite{arellanesmodels2025}. For example, in both HSMs and ONets, control is implicit in the activation of state transitions and, in Kahn networks, in data exchanges. Making control explicit in an MHC leads to formal reasoning of computation order, facilitating the verification of common computational properties such as reachability or termination \cite{vanderaalstsoundness2011}.

The computon model \cite{arellanescompositional2026,arellanesboolean2026} provides categorical semantics to formally compose high-level computing devices by control flow in an incremental and bottom-up manner. For this, it provides composition operators, with behaviour characterised as finite colimit constructions, that define explicit control flow for the activation of computing devices in some well-defined order. Data flow is considered, albeit it is a second-class dimension governed by control. Particularly, there are separate operators for forming sequential, (synchronous and asynchronous) parallel, branching and iterative computing devices. As their internal structure is only accessible through an interface, devices are modular black boxes with explicit separation of concerns. Separating data from control has proven relevant for reasoning about such dimensions independently \cite{clarkemodel2008}. For example, one can verify termination by analysing control flow only or data reachability without inspecting control flow at all \cite{vanderaalstsoundness2011}. Beyond verification, separating concerns has been effective for model transformation purposes \cite{arellanescompositional2026}.

Although the separation of control and data is a distinctive property of the computon model, such a model is still a generic formalism that serves as a reference rather than a concrete realisation dictating how devices should compute. Accordingly, it leaves the interpretation of computing devices open by avoiding concrete implementations and giving operational semantics at a higher-level of abstraction via P/T Petri nets. Moreover, the operators it provides do not offer sufficient flexibility for modelling expressive decision-making structures.

In this paper, we move from the abstract to the particular by endowing computons with structural relations between inputs and outputs to enable functional computing devices, and introduce a new operator for branching. We also show that the original operator for forming synchronous parallel composites can be defined out of sequencing and asynchronous parallelising, and study additional algebraic properties in terms of identity, a law that was not originally analysed in \cite{arellanescompositional2026,arellanesboolean2026}. The original definition of the sequential operator is modified to satisfy this law. 

Our ultimate goal is to offer a universal framework to formally reason about high-level computing devices in a compositional manner via categorical semantics. As the proposed framework is sufficiently general, we envision multiple theoretical extensions and tools built out of it. In this paper, we take the first step towards a tool that bridges theory and practice by implementing the computon model in a functional programming language with dependent types. The implementation yields an open-source programming environment, intended for describing high-level computations via colimit constructions. To the best of our knowledge, this is the first implementation of the computon model.

The rest of the paper is structured as follows. Section \ref{sec:category} presents the semantics of a category of computons and their morphisms. Section \ref{sec:operators} describes colimit constructions that capture the behaviour of control-driven composition operators for sequencing, parallelising and branching. Section \ref{sec:operation} defines concrete operational semantics for computon execution. Section \ref{sec:implementation} presents an implementation of the computon model. Section \ref{sec:related-work} discusses related work. Section \ref{sec:conclusions} outlines the conclusions and future work. 

\section{Categories of Computons and Computon Morphisms}
\label{sec:category}

A computon is intuitively a bipartite graph where nodes are ports or computation units, connected through edges that represent control flow or data flow. A port is a buffer for storing a datum or a control signal so it has a type associated to it. To abstract away from concrete types, a finite nonempty set ${\underline{n}:=[0,n)\cap\mathbb{N}}$ is used, capturing the essence of a fixed type system with ${n\geq 1}$ possible port types. In any choice, $0$ represents the type of control signals whereas the other numbers are used for other data types. Having $\underline{1}:=\{0\}$ as the minimal choice entails that control ports are always present whereas data ports are optional, i.e., computation is necessarily driven by control signals (cf. Fig. \ref{fig:examples-intro}c). Although computation units are not mandatory, as shown in Fig. \ref{fig:examples-intro}a, they always have control ports attached when they exist. In general, Fig. \ref{fig:examples-intro} shows diverse computon examples over the set $\underline{7}$.
\vspace{-16pt}
\begin{figure}[H]
\subcaptionbox{}
{
\begin{tikzpicture}[scale=0.76]
\qmatch{0t}{0}{2}{};
\qmatch{1t}{0}{1.6}{};
\qmatch{2t}{0}{1.2}{};
\dmatch{3t}{0}{0.8}{$5$}{left};
\dmatch{4t}{0}{0.4}{$4$}{left};
\dmatch{5t}{0}{0}{$5$}{left};
\end{tikzpicture}
}
\subcaptionbox{}
{
\begin{tikzpicture}[scale=0.76]
\primitive{0.5}{0}{0.6}{1.4}
\qin{i0pc}{0}{1.3}{};
\din{i1pc}{0}{0.9}{$6$};
\din{i1pc}{0}{0.5}{$2$};
\din{i2pc}{0}{0.1}{$3$};

\node at (0.9,0.7){\scriptsize $\nabla_4$};\draw[densely dotted] (0.5,0.9) -- (0.65,0.75);\draw[densely dotted] (0.5,0.5) -- (0.65,0.75);\draw[densely dotted] (0.5,0.1) -- (0.65,0.75);\draw[densely dotted] (0.9,0.8) -- (1.1,0.8);\draw[densely dotted] (0.5,1.3) -- (0.65,0.75);

\qout{o0pc}{1.1}{0.8}{};
\end{tikzpicture}
}
\subcaptionbox{}
{
\begin{tikzpicture}[scale=0.76]
\primitive{0.5}{0.3}{0.5}{1}
\qin{i0pc}{0}{0.8}{};
\qout{o0pc}{1}{0.8}{};

\node at (0.75,0.8){\scriptsize $\epsilon$};\draw[densely dotted] (0.5,0.8) -- (0.65,0.8);\draw[densely dotted] (0.85,0.8) -- (1,0.8);
\node at (0,0){};
\end{tikzpicture}
}
\subcaptionbox{}
{
\begin{tikzpicture}[scale=0.76]
\computonComposite{0.15}{-0.3}{2.75}{1.5};

\primitive{0.5}{0}{0.5}{1}
\qin{i0pc}{0}{0.9}{};
\qin{i1pc}{0}{0.3}{};
\node at (0.75,0.3){\tiny $9$};\draw[densely dotted] (0.5,0.3) -- (0.65,0.3);\draw[densely dotted] (0.87,0.3) -- (1,0.3);
\node at (0.75,0.85){\scriptsize $\epsilon$};\draw[densely dotted] (0.5,0.85) -- (0.65,0.85);\draw[densely dotted] (0.75,0.85) -- (1,0.85);

\qmatch{o0pc}{1.5}{0.85}{};\flow{{1,0.85}}{o0pc}{densely dashed}{};\flow{o0pc}{{2,0.85}}{densely dashed}{};
\doutplain{01pc}{2.5}{-0.15}{$1$};\flowdiag{{1,0.3}}{01pc}{bend right=15}{}{pos=0.3,rotate=-10};

\primitive{2}{0}{0.5}{1}
\qout{o0ac}{2.5}{0.85}{};
\node at (2.25,0.85){\scriptsize $\epsilon$};\draw[densely dotted] (2,0.85) -- (2.15,0.85);\draw[densely dotted] (2.35,0.85) -- (2.5,0.85);

\end{tikzpicture}
}
\subcaptionbox{}
{
\begin{tikzpicture}[scale=0.76]
\computonComposite{0.15}{-0.3}{2.75}{1.8};

\primitive{0.5}{0}{0.5}{1.4}
\din{i1pc}{0}{0.5}{$1$};
\din{i2pc}{0}{0.1}{$1$};
\node at (0.75,0.3){\scriptsize $\times$};\draw[densely dotted] (0.5,0.5) -- (0.65,0.3);\draw[densely dotted] (0.5,0.1) -- (0.65,0.3);\draw[densely dotted] (0.87,0.3) -- (1,0.3);
\node at (0.75,0.9){\tiny $\nabla$};\draw[densely dotted] (0.5,0.9) -- (0.65,0.9);\draw[densely dotted] (0.85,0.9) -- (1,0.9);
\node at (0.75,1.25){\scriptsize $\epsilon$};\draw[densely dotted] (0.5,1.25) -- (0.65,1.25);\draw[densely dotted] (0.85,1.25) -- (1,1.25);
\din{i3pc}{0}{0.9}{$6$};
\qin{i0pc}{0}{1.25}{};

\qmatch{o2pc}{1.5}{1.25}{};\flow{{1,1.25}}{o2pc}{densely dashed}{};\flow{o2pc}{{2,1.25}}{densely dashed}{};
\qmatch{o0pc}{1.5}{0.9}{};\flow{{1,0.9}}{o0pc}{densely dashed}{};\flow{o0pc}{{2,0.9}}{densely dashed}{};
\dmatch{o1pc}{1.5}{0.3}{$1$}{below};\flow{{1,0.3}}{o1pc}{}{};\flow{o1pc}{{2,0.3}}{}{};

\primitive{2}{0}{0.5}{1.4}
\qout{o0ac}{2.5}{1.25}{};
\dout{o1ac}{2.5}{0.3}{$1$};
\dout{o2ac}{2.5}{0.9}{$1$};
\node at (2.25,0.3){\scriptsize $S$};\draw[densely dotted] (2,0.3) -- (2.15,0.3);\draw[densely dotted] (2.37,0.3) -- (2.5,0.3);
\node at (2.25,0.9){\tiny $5$};\draw[densely dotted] (2,0.9) -- (2.15,0.9);\draw[densely dotted] (2.37,0.9) -- (2.5,0.9);
\node at (2.25,1.25){\scriptsize $\epsilon$};\draw[densely dotted] (2,1.25) -- (2.15,1.25);\draw[densely dotted] (2.35,1.25) -- (2.5,1.25);
\end{tikzpicture}
}
\subcaptionbox{}
{
\begin{tikzpicture}[scale=0.84]
\computonComposite{0.1}{-0.3}{1.3}{2.8};
\primitive{0.5}{1.2}{0.5}{1}
\din{i1p}{0}{2.05}{$1$};
\din{i2p}{0}{1.7}{$1$};
\qin{i0p}{0}{1.35}{};
\dout{o1p}{1}{1.9}{$1$};
\qout{o0p}{1}{1.35}{};
\node at (0.75,1.9){\scriptsize $\times$};\draw[densely dotted] (0.5,2.05) -- (0.65,1.9);\draw[densely dotted] (0.5,1.7) -- (0.65,1.85);\draw[densely dotted] (0.87,1.9) -- (1,1.9);
\node at (0.75,1.35){\scriptsize $\epsilon$};\draw[densely dotted] (0.5,1.35) -- (0.65,1.35);\draw[densely dotted] (0.85,1.35) -- (1,1.35);

\primitive{0.5}{0}{0.5}{1}
\qin{i0a}{0}{0.85}{};
\din{i1a}{0}{0.5}{$1$};
\din{i2a}{0}{0.1}{$2$};
\qout{o0a}{1}{0.85}{};
\dout{o1a}{1}{0.3}{$2$};
\node at (0.75,0.3){\scriptsize $+$};\draw[densely dotted] (0.5,0.5) -- (0.65,0.3);\draw[densely dotted] (0.5,0.1) -- (0.65,0.3);\draw[densely dotted] (0.87,0.3) -- (1,0.3);
\node at (0.75,0.85){\scriptsize $\epsilon$};\draw[densely dotted] (0.5,0.85) -- (0.65,0.85);\draw[densely dotted] (0.85,0.85) -- (1,0.85);
\end{tikzpicture}
}
\subcaptionbox{}
{
\begin{tikzpicture}[scale=0.84]
\computonComposite{0.2}{-0.2}{1.1}{2.6};

\primitive{0.5}{1.2}{0.5}{1}
\node at (0.75,1.5){\scriptsize $S$};\draw[densely dotted] (0.5,1.5) -- (0.65,1.5);\draw[densely dotted] (0.87,1.5) -- (1,1.5);
\node at (0.75,2.05){\scriptsize $\epsilon$};\draw[densely dotted] (0.5,2.05) -- (0.65,2.05);\draw[densely dotted] (0.85,2.05) -- (1,2.05);

\primitive{0.5}{0}{0.5}{1}
\node at (0.75,0.3){\scriptsize $!$};\draw[densely dotted] (0.5,0.3) -- (0.65,0.3);\draw[densely dotted] (0.87,0.3) -- (1,0.3);
\node at (0.75,0.85){\scriptsize $\epsilon$};\draw[densely dotted] (0.5,0.85) -- (0.65,0.85);\draw[densely dotted] (0.85,0.85) -- (1,0.85);

\qinplain{i0}{0}{1.4}{};\flowdiag{i0}{{0.5,2.05}}{dashed}{}{pos=0.5,rotate=45};\flowdiag{i0}{{0.5,0.85}}{dashed}{}{pos=0.7,rotate=315};	
\dinplain{i1}{0}{1}{$1$}{left};\flowdiag{i1}{{0.5,1.5}}{}{}{pos=0.7,rotate=45};\flowdiag{i1}{{0.5,0.3}}{}{}{pos=0.5,rotate=315};
\qoutplain{o0}{1}{1.4}{};\flowdiag{{1,2.05}}{o0}{dashed}{}{pos=0.5,rotate=315};\flowdiag{{1,0.85}}{o0}{dashed}{}{pos=0.3,rotate=45};	
\doutplain{o1}{1}{1}{$1$}{right};\flowdiag{{1,1.5}}{o1}{}{}{pos=0.3,rotate=315};\flowdiag{{1,0.3}}{o1}{}{}{pos=0.6,rotate=45};
\end{tikzpicture}
}
{
\begin{center}
\begin{tikzpicture}
\begin{scope}
\draw[draw=black,fill=white] (0,0.4) rectangle ++(0.12,0.25);
\node[inner sep=0pt,minimum size=0pt,label=right:{\scriptsize Comp. unit}] at (0.2,0.5) {};
\draw[densely dashed] (2.1,0.5) to node[pos=0.5]{\arrowflow} (2.6,0.5);
\node[inner sep=0pt,minimum size=3pt,label=right:{\scriptsize Control flow}] at (2.5,0.5) {};
\node[draw=black,fill=white,inner sep=0pt,minimum size=3pt,label=right:{\scriptsize Control inport}] at (4.8,0.5) {};
\node[fill=black,inner sep=0pt,minimum size=3pt,label={right:{\scriptsize Control outport}}] at (7.3,0.5) {};
\node[draw,fill=white,fill fraction={black}{0.5},inner sep=0pt,minimum size=3pt,label=right:{\scriptsize Control inoutport}] at (9.8,0.5) {};
\end{scope}

\begin{scope}
\draw (0,0) to node[pos=0.5]{\arrowflow} (0.4,0);
\node[inner sep=0pt,minimum size=3pt,label=right:{\scriptsize Data flow}] at (0.3,0){};
\node[circle,draw=black,fill=white,inner sep=0pt,minimum size=3pt,label={right:{\scriptsize Data inport}}] at (2.1,0) {};
\node[circle,fill=black,inner sep=0pt,minimum size=3pt,label={right:{\scriptsize Data outport}}] at (4.2,0) {};
\node[circle,draw,fill=white,fill fraction={black}{0.5},inner sep=0pt,minimum size=3pt,label=right:{\scriptsize Data inoutport}] at (6.5,0) {};
\draw[draw=black,fill=black!2] (9.1,-0.05) rectangle ++(0.3,0.2);
\node[inner sep=0pt,minimum size=3pt,label=right:{\scriptsize Composite computon}] at (9.35,0){};
\end{scope}
\end{tikzpicture}
\end{center}
}
\vspace{-6pt}
\caption{Examples of computons over the set $\underline{7}$ of types and the set $\{\nabla_4,\epsilon,9,\nabla,\times,5,S,+,!\}$ of computing devices which, for clarity, are not displayed as strings but as symbols expressing behaviour. Particularly, the devices in this example correspond to programs for discarding 4 values ($\nabla_4$), echoing a single control signal ($\epsilon$), producing the constant $9$, discarding a single value ($\nabla$), computing binary multiplication ($\times$), producing the constant $5$, computing the successor of a number ($S$), adding two numbers ($+$) and computing the factorial of a number ($!$). Computons abstract away from concrete types in order to provide a general implementation-independent framework, e.g., $1$ can express the type of integers, $2$ the type of floats and so on. In Sect. \ref{sec:operation}, we describe operational semantics to map natural numbers to concrete types from a fixed type system. Unlike data ports, control ports are never $0$-labelled for clarity.} 
\label{fig:examples-intro}
\end{figure}
\vspace{-9pt}
Intuitively, ports are attached to units to indicate the type of values a unit can receive/produce. Rather than denoting a single computational behaviour, a unit is a collection of computing devices (i.e., programs) that receive at least one input and produce exactly one output. Particularly, the inputs of a device $d$ in a unit $u$ come from a subset of ports connected to $u$, whereas the output of $d$ is sent a port connected from $u$. The set of all possible computing devices is a nonempty finite set ${B\subset\Sigma^*}$, i.e., a device is a finite string over some finite alphabet $\Sigma$.\footnote{Representing a computing device as a finite sequence of symbols over a finite alphabet is valid since it is well-known that every program (including computable functions) can be encoded thereby \cite{sipserintroduction2013}.} One of the particular properties is that $B$ always contains a device ${\epsilon}$ that echoes a single control signal, as shown in Fig. \ref{fig:examples-intro}c. As there evidently are multiple choices for $B$ and $\underline{n}$, one can form multiple computon categories. For a particular choice, objects are defined as follows.

\begin{definition}[Computon]\label{def:computon}
A computon $\lambda$ is a 13-tuple $(U,P,I,O,\underline{n},B,\sigma,t,\tau,s,c,r,f)$ where:
\begin{itemize}[noitemsep, topsep=0pt]
\item $U$ is a finite set of computation units,
\item $P$ is a finite nonempty set of ports,
\item $I$ is a finite set of inflows,
\item $O$ is a finite set of outflows,
\item ${\sigma\colon O\twoheadrightarrow U}$ is a surjective function that defines the source unit of each outflow,
\item ${t\colon O\to P}$ is a function that specifies the target port of each outflow,
\item ${\tau\colon I\twoheadrightarrow U}$ is a surjective function that specifies the target unit of each inflow,
\item ${s\colon I\to P}$ is a function that specifies the source port of each inflow,
\item ${c\colon P\to \underline{n}}$ is a function that assigns a type to each port, 
\item ${r\colon I\twoheadrightarrow O}$ is a surjective function relating inflows with outflows, and
\item ${f\colon O\to B}$ is a function that attaches each outflow to a computing device
\end{itemize}
such that
\begin{enumerate*}[label=(\roman*)]
\item ${0\in\underline{n}}$,\label{def:computon-1}
\item ${\epsilon\in B}$,\label{def:computon-2}
\item ${\tau\restriction_{(c\circ s)^{-1}(0)}}$ and ${\sigma\restriction_{(c\circ t)^{-1}(0)}}$ are surjective,\label{def:computon-3}
\item ${\sigma\circ r = \tau}$, and \label{def:computon-4}
\item there are ports ${p \in P\setminus t(O)}$ and ${q \in P\setminus s(I)}$ with ${c(p)=0=c(q)}$.\label{def:computon-5}
\end{enumerate*}
\end{definition}

\begin{notation}\label{not:pre-post-sets}
A computon $\lambda$ with ${u \in U}$ and ${p\in P}$ has sets ${s(\tau^{-1}(u))}$, ${t(\sigma^{-1}(u))}$, ${\tau(s^{-1}(p))}$ and ${\sigma(t^{-1}(p))}$ written ${\bullet u}$, ${u\bullet}$, ${p\bullet}$ and ${\bullet p}$, respectively. To distinguish between computons, we use natural numbers to index their components. If a computon symbol has no subindex, its components have no subindex either. 
\end{notation}

Rather than directly specifying a function ${U\to B}$, Definition \ref{def:computon} relies on a span ${U\overset{\sigma}{\twoheadleftarrow}O\xrightarrow{f}B}$ to generalise the functional relation given by the former, i.e., a unit can be related not to just one device, but to multiple ones. Having ${U\to B}$ alone is not sufficient to allow units produce different outputs consistently, as such a function does not encode any means to determine which values go to which outflows without imposing certain order, i.e., the only way of enabling multiple outputs is through a product type. By equipping computons with a function chain ${I\overset{\text{r}}{\twoheadrightarrow}O\xrightarrow{f}B}$, we abstract away from particular orderings by offering a direct relationship between inflows and outflows that makes it possible to extract information flows coming from/into computing devices. Concretely, an outflow ${o\in O}$ takes data from the result of a device ${f(o)}$ which, in turn, operates on the input values from ${r^{-1}(o)}$. By the totality of $f$, it is guaranteed that each outflow reads information from exclusively one device. 

To encapsulate behaviour within units, Restriction \ref{def:computon-4} in Definition \ref{def:computon} enforces devices to read/produce data from/into ports attached to the unit they belong to. Restriction \ref{def:computon-3} enforces units to have control ports connected to and from it, and Restrictions \ref{def:computon-1} and \ref{def:computon-2} respectively specify that the control type and the signal echoing device $\epsilon$ are ever present. Restriction \ref{def:computon-5} ensures the existence of ports where control flow starts and terminates, called control inports and control outports, respectively. The set ${P^+}$ of all control and data inports defines the input interface of a computon $\lambda$, whilst the set ${P^-}$ of all control and data outports give rise to the output interface. 

\begin{definition}[Computon Interface]\label{def:interface}
The interface of a computon $\lambda$ is a tuple ${(P^+,P^-)}$ where ${P^+}$ and ${P^-}$ are the sets ${P\setminus t(O)}$ and ${P\setminus s(I)}$, respectively.
\end{definition}

When there is a sequence of flows from every inport (or from any port with outflows) to some outport, we say that a computon is connected (see Definition \ref{def:connected} and Fig. \ref{fig:examples-intro}e). By Proposition \ref{prop:connected-units}, connected computons always have at least one computation unit. 

\begin{definition}[Connected Computon]\label{def:connected}
Let ${(U \sqcup P, I \sqcup O, s', t')}$ be the bipartite graph $G$ of a computon $\lambda$ with ${s',t'\colon I\sqcup O\to U\sqcup P}$ given as follows: $\begin{aligned}
  s'(e) &=
    \begin{cases}
      s(e) & \text{if } e\in I \\
      \sigma(e) & \text{if } e\in O
    \end{cases}
  \qquad\qquad
  t'(e) &=
    \begin{cases}
      \tau(e) & \text{if } e\in I \\
      t(e) & \text{if } e\in O
    \end{cases}
\end{aligned}$
We say that a computon is connected if, for every port ${p\in P^+\cup s(I)}$, there is a path ${(e_1,e_2,\ldots,e_n)}$ in $G$ of length ${n\geq 2}$ with ${s'(e_1)=p}$ and ${t'(e_n)\in P^-}$.
\end{definition}

\begin{proposition}\label{prop:connected-units}
If $\lambda$ is a connected computon, then ${U\neq\emptyset}$.
\end{proposition}

A computon with only interface ports and no units or flows at all is called a \emph{trivial computon} (see Definition \ref{def:trivial-computon} and Fig. \ref{fig:examples-intro}a). A trivial computon with exactly one port is referred to as a \emph{unit computon}.

\begin{definition}\label{def:trivial-computon}
A computon $\lambda$ is trivial if it satisfies ${U=I=O=\emptyset}$. A trivial computon with ${|P|=1}$ is called a unit computon.
\end{definition}

\begin{remark}
By Definitions \ref{def:computon}, \ref{def:interface} and \ref{def:trivial-computon}, it is easy to see that a unit computon is an entity with exactly one port ${p \in P^+\cap P^-}$ with ${c(p)=0}$, i.e., it consists of only one control port that is inport and outport simultaneously. Ports in ${(P^+\cap P^-)\cup(s(I)\cap t(O))}$ are referred to as inoutports.
\end{remark}

When a computon has exactly one computation unit to which all ports are attached, we say it is \emph{primitive} (see Fig. \ref{fig:examples-intro}b). Like trivials, a primitive has all ports at the interface, with the restriction that each inport is connected to the unique unit via a single inflow and outports are linked to the unit through a single outflow. This is enforced by the injectivity condition on the functions $s$ and $t$ in Definition \ref{def:primitive-computon}.

\begin{definition}\label{def:primitive-computon}
A computon $\lambda$ is primitive if ${|U|=1}$, ${P\cong I\sqcup O}$ and $s$ and $t$ are injective.
\end{definition}

Note that Definition \ref{def:computon} enforces primitive computons to always have inflows and outflows due to the surjectivity of $\tau$ and $\sigma$. As $s$ and $t$ are not necessarily onto in that definition, it is possible to have ports with no flows attached. To prevent this, ${P\cong I\sqcup O}$ and injectivity over $s$ and $t$ are required, with the former condition implying that all ports lie at the interface (see Proposition \ref{prop:primitive-interface}).\footnote{In Proposition \ref{prop:primitive-interface}, we use the symbol $\triangle$ to denote symmetric set difference.}

\begin{proposition}\label{prop:primitive-interface}
If $\lambda$ is a primitive computon, $P\cong P^+\triangle P^-$.
\end{proposition}

\subsection{Computon Morphisms}

A morphism in a computon category is intuitively an insertion of a computon into another. More precisely, it is a collection of six total functions that map units to units, ports to ports, inflows to inflows, outflows to outflows, types to types and devices to devices such that the two last mappings are inclusions to guarantee consistency in terms of port typing and computational behaviour (see Definition \ref{def:morphism}). 

\begin{definition}[Computon Morphism]\label{def:morphism}
In a computon category over a set $\underline{n}$ of types and a set $B$ of computing devices, a morphism ${\alpha\colon\lambda_1\rightarrow\lambda_2}$ is a 6-tuple $(\alpha_U,\alpha_P,\alpha_I,\alpha_O,\alpha_{\underline{n}},\alpha_B)$ of total functions ${\alpha_U\colon U_1\rightarrow U_2}$, ${\alpha_P\colon P_1\rightarrow P_2}$, ${\alpha_I\colon I_1\rightarrow I_2}$, $\alpha_O\colon O_1\rightarrow O_2$, ${\alpha_{\underline{n}}\colon\underline{n}\hookrightarrow\underline{n}}$ and ${\alpha_B\colon B\hookrightarrow B}$ such that ${\vec{i}(\alpha)\cup\vec{o}(\alpha)\subseteq P_1^+\cup P_1^-}$ and the following diagrams commute:

\hspace*{-20pt}
\begin{tikzcd}
I_1 \arrow[r, twoheadrightarrow, "\tau_1"]\arrow[d, "\alpha_I"']
& U_1 \arrow[d, "\alpha_U"] 
& O_1 \arrow[l, twoheadrightarrow, "\sigma_1"']\arrow[d, "\alpha_O"] \\
I_2 \arrow[r, twoheadrightarrow, "\tau_2"']
& U_2 
& O_2 \arrow[l, twoheadrightarrow, "\sigma_2"]
\end{tikzcd}
\quad
\begin{tikzcd}
I_1 \arrow[r, "s_1"]\arrow[d, "\alpha_I"']
& P_1 \arrow[d, "\alpha_P"] 
& O_1 \arrow[l, "t_1"']\arrow[d, "\alpha_O"] \\
I_2 \arrow[r, "s_2"']
& P_2 
& O_2 \arrow[l, "t_2"]
\end{tikzcd}
\begin{tikzcd}
B_1 \arrow[d, hook, "\alpha_B"'] & O_1 \arrow[d, "\alpha_O"]\arrow[l,"f_1"'] & I_1 \arrow[d, "\alpha_I"]\arrow[l, twoheadrightarrow, "r_1"'] \\
B_2 & O_2 \arrow[l,"f_2"] & I_2 \arrow[l, twoheadrightarrow, "r_2"]
\end{tikzcd}
\quad
\begin{tikzcd}
P_1 \arrow[r, "c_1"]\arrow[d, "\alpha_P"']
& \underline{n} \arrow[d, hook, "\alpha_{\underline{n}}"] \\
P_2 \arrow[r, "c_2"']
& \underline{n}
\end{tikzcd}

Here, ${\vec{i}(\alpha)}$ and ${\vec{o}(\alpha)}$ are given by ${\{p\in P_1\mid\bullet\alpha_P(p)\setminus\alpha_P(\bullet p)\neq\emptyset\}}$ and ${\{p\in P_1\mid\alpha_P(p)\bullet\setminus\alpha_P(p\bullet)\neq\emptyset\}}$, respectively. From now on, we use natural numbers as subindices to distinguish between morphisms and abuse notation by omitting components, e.g., we write ${\alpha(p)}$ for ${\alpha_P(p)}$ when the context is clear. 
\end{definition}

The leftmost diagram in Definition \ref{def:morphism} ensures preservation of inflow and outflow adjacency with respect to computation units. The immediate diagram on the right enforces preservation of the adjacency between ports and their information flows. The isolated square ensures that port typing is retained, and the remaining one keeps the relationship between inflows, outflows and devices. The restriction ${\vec{i}(\alpha)\cup\vec{o}(\alpha)\subseteq P_1^+\cup P_1^-}$ ensures that a computon can be inserted into another only at its interface, i.e., only inports or outports of the source computon can be attached to new computation units in the target computon. As a consequence, we have the following lemma.

\begin{lemma}[\cite{arellanescompositional2026}]\label{prop:interface-preservation}
For any computon morphism ${\alpha\colon\lambda_1\rightarrow\lambda_2}$, ${\alpha^{-1}(P_2^+)\subseteq P_1^+}$ and ${\alpha^{-1}(P_2^-)\subseteq P_1^-}$.
\end{lemma}

There is a special class of morphisms, called markers, which allow embedding a trivial computon into a computon interface (see Definition \ref{def:markers}). When the embedding covers all the inports, it is called an \emph{in-marker} and, when it covers all the outports, it is called an \emph{out-marker}. 

\begin{definition}[Computon Markers]\label{def:markers}
A marker ${\lambda^\square}$ of a computon $\lambda$ is a computon monomorphism ${\lambda_0\rightarrow\lambda}$ where $\lambda_0$ is a trivial computon and ${\lambda^\square(P_0)=P^\square}$ with ${\square\in\{+,-\}}$. If ${\square=+}$, it is called an \emph{in-marker}; otherwise, an \emph{out-marker}.
\end{definition}

Interestingly, there are markers whose domain is a unit computon, indicating that the codomain does not require (or produce) any data but just receives (or produces) a single control signal. Although a unit computon evidently embodies the minimal structure one can establish from Definition \ref{def:computon}, such an entity is not an initial object in a computon category since there are $k$ morphisms from it to a computon with $k$ control ports, rather than a unique morphism, i.e., a computon category has no initial objects (cf. \cite{arellanescompositional2026}).

\subsection{Colimit Constructions}

In the theory of computons, complex embeddings can be built out of ``elementary'' finite colimits, namely coproduct and pushout. Coproduct is simply the disjoint union of computon components, which gives rise to a composite structure that puts two computons side-by-side (see Definition \ref{def:coproduct}). Pushout is the square complement of a span of computon morphisms, which merges two computons into a composite according to the instructions given by the span (see Definition \ref{def:pushout}). In \cite{arellanescompositional2026}, we showed that pushouts can only be formed if a span satisfies the restrictions from Definition \ref{def:pushable-span} and that a category of computons has all coproducts (see Theorem \ref{th:pushable-pushout-cat}). In the same work, we also showed that both coproduct and pushout satisfy the required universal properties in a computon category.
\begin{definition}[Coproduct]\label{def:coproduct}
In a computon category over a set $\underline{n}$ of types and a set $B$ of computing devices, the coproduct ${\lambda_1+\lambda_2}$ of $\lambda_1$ and $\lambda_2$ is given by the following diagram of finite sets and total functions:
\begin{center}
\begin{tikzcd}[ampersand replacement=\&]
 \& O_1\sqcup O_2 \arrow[dl,twoheadrightarrow,"\sigma"']\arrow[dr, "t"]\arrow[r, "f"] \& B \&  \\
U_1\sqcup U_2 \& I_1\sqcup I_2 \arrow[l,twoheadrightarrow,"\tau"]\arrow[r, "s"']\arrow[u,twoheadrightarrow,"r"'] \& P_1\sqcup P_2 \arrow[r,"c"] \& \underline{n}
\end{tikzcd}
\end{center}
Supposing ${\beta_1\colon\lambda_1\to\lambda_1+\lambda_2}$ and ${\beta_2\colon\lambda_2\to\lambda_1+\lambda_2}$ are the canonical coproduct monomorphisms, the functions $\sigma$, $t$, $\tau$, $s$, $c$, $r$ and $f$ are computed sourcewise. For example, $\sigma$ is given as follows for all ${o\in O_1\sqcup O_2}$: $\sigma(o)=
    \begin{cases}
        \beta_1(\sigma_1(o_1)) & \text{if } o=\beta_1(o_1) \text{ for some }o_1\in O_1\\
        \beta_2(\sigma_2(o_2)) & \text{if } o=\beta_2(o_2) \text{ for some }o_2\in O_2
    \end{cases}$
\end{definition}

\begin{definition}[Pushout]\label{def:pushout}
In a computon category over a set $\underline{n}$ of types and a set $B$ of computing devices, the pushout ${\lambda_1+_{\lambda_0}\lambda_2}$ of a span ${\lambda_1\xleftarrow{\alpha_1}\lambda_0\xrightarrow{\alpha_2}\lambda_2}$ of computon morphisms is given by the following diagram of finite sets and total functions:

\begin{center}
\begin{tikzcd}[ampersand replacement=\&]
 \& O_1\sqcup_{O_0} O_2 \arrow[dl,twoheadrightarrow,"\sigma"']\arrow[dr, "t"]\arrow[r, "f"] \& B \&  \\
U_1\sqcup_{U_0} U_2 \& I_1\sqcup_{I_0} I_2 \arrow[l,twoheadrightarrow,"\tau"]\arrow[r, "s"']\arrow[u,twoheadrightarrow,"r"'] \& P_1\sqcup_{P_0} P_2 \arrow[r,"c"] \& \underline{n}
\end{tikzcd}
\end{center}
Supposing ${\beta_1\colon\lambda_1\to\lambda_1+_{\lambda_0}\lambda_2}$ and ${\beta_2\colon\lambda_2\to\lambda_1+_{\lambda_0}\lambda_2}$ are the pushout-induced computon morphisms, the functions $\sigma$, $t$, $\tau$, $s$, $c$, $r$ and $f$ are computed sourcewise. For example, $\tau$ is given as follows for all ${i\in I_1\sqcup_{I_0} I_2}$: $\tau(i)=
    \begin{cases}
        \beta_1(\tau_1(i_1)) & \text{if } i=\beta_1(i_1) \text{ for some }i_1\in I_1\\
        \beta_2(\tau_2(i_2)) & \text{if } i=\beta_2(i_2) \text{ for some }i_2\in I_2
    \end{cases}$
\end{definition}

\begin{definition} [Pushable Span] \label{def:pushable-span}
A span ${\lambda_1 \xleftarrow{\alpha_1} \lambda_0 \xrightarrow{\alpha_2} \lambda_2}$ is pushable if ${\alpha_1(\vec{i}(\alpha_2)) \cup \alpha_1(\vec{o}(\alpha_2)) \subseteq P_1^+ \cup P_1^-}$ and ${\alpha_2(\vec{i}(\alpha_1)) \cup \alpha_2(\vec{o}(\alpha_1)) \subseteq P_2^+ \cup P_2^-}$.
\end{definition}

\begin{remark}\label{remark:union}
Note that when ${i\in I_1\sqcup_{I_0}I_2}$ is identified with elements ${i_1\in I_1}$ and ${i_2\in I_2}$, it is sufficient for $\tau$ to choose either ${\beta_1(\tau_1(i_1))}$ or ${\beta_2(\tau_2(i_2))}$ due to the commutativity equations from Definition \ref{def:morphism}. The same applies to $\sigma$, $s$, $t$, $c$, $r$ and $f$.
\end{remark}

\begin{theorem}[\cite{arellanescompositional2026}]\label{th:pushable-pushout-cat}
Let $\lambda_1\xleftarrow{\alpha_1}\lambda_0\xrightarrow{\alpha_2}$ be a span $\rho$ of computon morphisms. The pushout of $\rho$ exists if and only if $\rho$ is pushable. If $\alpha_1$ and $\alpha_2$ are computon markers, $\rho$ is pushable.
\end{theorem}

\section{Composition Operators}
\label{sec:operators}

Pushouts and coproducts form the basis of the composition operators the theory of computons builds upon. In this section, we describe separate operators to form sequential, parallel and branching composites, which define explicit control flow for the invocation of computons in some precise order. 

\vspace{-10pt}

\subsection{Sequential Computons}

A sequential computon defines a control flow structure for the invocation of two computons in a pipeline. Such a composite is obtained by computing the pushout of a sequentiable span of computon morphisms (see Definition \ref{def:sequentiable-span}). It is total when all the outports of the ``first-executed'' computon are connected to all the inports of the ``secondly-executed'' computon, and partial otherwise (see Definition \ref{def:sequential-computon}). 

\begin{definition}[Sequentiable Span]\label{def:sequentiable-span}
A span ${\lambda_1\xleftarrow{\alpha_1}\lambda_0\xrightarrow{\alpha_2}\lambda_2}$ is sequentiable if $\lambda_0$ is a trivial computon, and $\alpha_1$ and $\alpha_2$ are monomorphisms with ${\alpha_1(P_0)\subseteq P_1^-}$ and ${\alpha_2(P_0)\subseteq P_2^+}$. 
\end{definition}

\begin{theorem}[\cite{arellanescompositional2026}]\label{th:sequentiable-is-pushable}
Every sequentiable span is pushable.
\end{theorem}

\begin{definition}[Sequential Computon]\label{def:sequential-computon}
The pushout of a sequentiable span ${\rho:=\lambda_1\xleftarrow{\alpha_1}\lambda_0\xrightarrow{\alpha_2}\lambda_2}$ yields a total sequential computon ${\lambda_1\unrhd_{\rho}\lambda_2}$ if ${\alpha_1(P_0)=P_1^-}$ and ${\alpha_2(P_0)=P_2^+}$; otherwise, it yields a partial sequential computon ${\lambda_1\rhd_{\rho}\lambda_2}$.
\end{definition}

\begin{example}
Suppose we have primitive computons $\lambda_1$, $\lambda_2$ and $\lambda_3$ with units that encapsulate devices for binary multiplication ($\times$), binary addition ($+$) and unary successor ($S$), respectively, each also having the device $\epsilon$. With these primitives, we can form composites in diverse ways. For example, Fig. \ref{fig:example-sequencing}a shows the construction of a partial sequential computon ${\lambda_1\rhd_{\rho_1}\lambda_2}$, intended for computing ${(a\times b)+c}$. Partiality occurs because this composite is formed from a sequentiable span ${\rho_1:=(\alpha_1,\alpha_2)}$ that omits the $2$-coloured data inport of $\lambda_2$. An example for the formation of a total sequential computon ${\lambda_1\unrhd_{\rho_2}\lambda_3}$ is shown in Fig. \ref{fig:example-sequencing}b, in which there is a sequentiable span ${\rho_2:=(\alpha_1,\alpha_3)}$ that identifies all the $\lambda_1$-outports with all the $\lambda_3$-inports, yielding a composite intended to compute the successor of ${a\times b}$. If $1$ and $2$ are the respective types of nonnegative integers and float numbers, then ${a,b\geq 0}$, and $c$ would be a float. 
\end{example}
\vspace{-15pt}
\begin{figure}[H]
\centering
\subcaptionbox{Constructing ${\lambda_1\rhd_{\rho_1}\lambda_2}$.}
{
\begin{tikzpicture}[scale=0.76]
\draw[thick,-{Stealth[inset=0pt, length=3pt, angle'=90]},opacity=0.4] (5.1,2.5) to node[above]{\scriptsize$\alpha_1$} (2,2.5);
\draw[thick,-{Stealth[inset=0pt, length=3pt, angle'=90]},opacity=0.4] (5.7,2) to node[right]{\scriptsize$\alpha_2$} (5.7,1.2);
\draw[thick,-{Stealth[inset=0pt, length=3pt, angle'=90]},opacity=0.4] (4.5,0.6) to node[above]{} (3.5,0.6);
\draw[thick,-{Stealth[inset=0pt, length=3pt, angle'=90]},opacity=0.4] (0.75,1.9) to node[left]{} (0.75,1.2);
\begin{scope}[yshift=2cm]
\primitive{0.5}{0}{0.5}{1}
\qin{i0p}{0}{0.85}{};
\din{i1p}{0}{0.5}{$1$};
\din{i2p}{0}{0.1}{$1$};
\qout{o0p}{1}{0.85}{};
\dout{o1p}{1}{0.3}{$1$};
\node at (0.75,0.3){\scriptsize $\times$};\draw[densely dotted] (0.5,0.5) -- (0.65,0.3);\draw[densely dotted] (0.5,0.1) -- (0.65,0.3);\draw[densely dotted] (0.87,0.3) -- (1,0.3);
\node at (0.75,0.85){\scriptsize $\epsilon$};\draw[densely dotted] (0.5,0.85) -- (0.65,0.85);\draw[densely dotted] (0.85,0.85) -- (1,0.85);
\node[circle,fill=blue,inner sep=0pt,minimum size=7pt,opacity=0.2] at (o0p) {};
\node[circle,fill=orange,inner sep=0pt,minimum size=7pt,opacity=0.2] at (o1p) {};
\node[opacity=0.4] at (0.75,1.25){\scriptsize $\lambda_1$};
\end{scope}

\begin{scope}[xshift=5.7cm,yshift=2cm]
\qmatch{0t}{0}{0.7}{};
\dmatch{1t}{0}{0.3}{$1$}{left};
\node[circle,fill=blue,inner sep=0pt,minimum size=7pt,opacity=0.2] at (0t) {};
\node[circle,fill=orange,inner sep=0pt,minimum size=7pt,opacity=0.2] at (1t) {};
\end{scope}

\begin{scope}[xshift=5cm]
\primitive{0.5}{0}{0.5}{1}
\qin{i0a}{0}{0.85}{};
\din{i1a}{0}{0.5}{$1$};
\din{i2a}{0}{0.1}{$2$};
\qout{o0a}{1}{0.85}{};
\dout{o1a}{1}{0.3}{$2$};
\node at (0.75,0.3){\scriptsize $+$};\draw[densely dotted] (0.5,0.5) -- (0.65,0.3);\draw[densely dotted] (0.5,0.1) -- (0.65,0.3);\draw[densely dotted] (0.87,0.3) -- (1,0.3);
\node at (0.75,0.85){\scriptsize $\epsilon$};\draw[densely dotted] (0.5,0.85) -- (0.65,0.85);\draw[densely dotted] (0.85,0.85) -- (1,0.85);
\node[circle,fill=blue,inner sep=0pt,minimum size=7pt,opacity=0.2] at (i0a) {};
\node[circle,fill=orange,inner sep=0pt,minimum size=7pt,opacity=0.2] at (i1a) {};
\node[opacity=0.4] at (0.75,-0.25){\scriptsize $\lambda_2$};
\end{scope}

\begin{scope}
\computonComposite{0.15}{-0.3}{2.75}{1.4};

\primitive{0.5}{0}{0.5}{1}
\qin{i0pc}{0}{0.85}{};
\din{i1pc}{0}{0.5}{$1$};
\din{i2pc}{0}{0.1}{$1$};
\node at (0.75,0.3){\scriptsize $\times$};\draw[densely dotted] (0.5,0.5) -- (0.65,0.3);\draw[densely dotted] (0.5,0.1) -- (0.65,0.3);\draw[densely dotted] (0.87,0.3) -- (1,0.5);
\node at (0.75,0.85){\scriptsize $\epsilon$};\draw[densely dotted] (0.5,0.85) -- (0.65,0.85);\draw[densely dotted] (0.85,0.85) -- (1,0.85);

\qmatch{o0pc}{1.5}{0.85}{};\flow{{1,0.85}}{o0pc}{densely dashed}{};\flow{o0pc}{{2,0.85}}{densely dashed}{};
\dmatch{o1pc}{1.5}{0.5}{$1$}{below};\flow{{1,0.5}}{o1pc}{}{};\flow{o1pc}{{2,0.5}}{}{};

\primitive{2}{0}{0.5}{1}
\qout{o0ac}{2.5}{0.85}{};
\dout{o1ac}{2.5}{0.3}{$2$};
\node at (2.25,0.3){\scriptsize $+$};\draw[densely dotted] (2,0.5) -- (2.15,0.3);\draw[densely dotted] (2,0.1) -- (2.15,0.3);\draw[densely dotted] (2.37,0.3) -- (2.5,0.3);
\node at (2.25,0.85){\scriptsize $\epsilon$};\draw[densely dotted] (2,0.85) -- (2.15,0.85);\draw[densely dotted] (2.35,0.85) -- (2.5,0.85);

\dinplain{i1ac}{0}{-0.2}{$2$};\flowdiag{i1ac}{{2,0.1}}{bend right=10}{}{pos=0.7,rotate=30};
\node[circle,fill=blue,inner sep=0pt,minimum size=7pt,opacity=0.2] at (o0pc) {};
\node[circle,fill=orange,inner sep=0pt,minimum size=7pt,opacity=0.2] at (o1pc) {};
\end{scope}
\end{tikzpicture}
}
\subcaptionbox{Constructing ${\lambda_1\unrhd_{\rho_2}\lambda_3}$.}
{
\begin{tikzpicture}[scale=0.76]
\draw[thick,-{Stealth[inset=0pt, length=3pt, angle'=90]},opacity=0.4] (5.1,2.5) to node[above]{\scriptsize$\alpha_1$} (2,2.5);
\draw[thick,-{Stealth[inset=0pt, length=3pt, angle'=90]},opacity=0.4] (5.7,2) to node[right]{\scriptsize$\alpha_3$} (5.7,1.2);
\draw[thick,-{Stealth[inset=0pt, length=3pt, angle'=90]},opacity=0.4] (4.5,0.6) to node[above]{} (3.5,0.6);
\draw[thick,-{Stealth[inset=0pt, length=3pt, angle'=90]},opacity=0.4] (0.75,1.9) to node[left]{} (0.75,1.2);
\begin{scope}[yshift=2cm]
\primitive{0.5}{0}{0.5}{1}
\qin{i0p}{0}{0.9}{};
\din{i1p}{0}{0.5}{$1$};
\din{i2p}{0}{0.1}{$1$};
\qout{o0p}{1}{0.85}{};
\dout{o1p}{1}{0.3}{$1$};
\node at (0.75,0.3){\scriptsize $\times$};\draw[densely dotted] (0.5,0.5) -- (0.65,0.3);\draw[densely dotted] (0.5,0.1) -- (0.65,0.3);\draw[densely dotted] (0.87,0.3) -- (1,0.3);
\node at (0.75,0.85){\scriptsize $\epsilon$};\draw[densely dotted] (0.5,0.85) -- (0.65,0.85);\draw[densely dotted] (0.85,0.85) -- (1,0.85);
\node[circle,fill=blue,inner sep=0pt,minimum size=7pt,opacity=0.2] at (o0p) {};
\node[circle,fill=orange,inner sep=0pt,minimum size=7pt,opacity=0.2] at (o1p) {};
\node[opacity=0.4] at (0.75,1.25){\scriptsize $\lambda_1$};
\end{scope}

\begin{scope}[xshift=5.7cm,yshift=2cm]
\qmatch{0t}{0}{0.7}{};
\dmatch{1t}{0}{0.3}{$1$}{left};
\node[circle,fill=blue,inner sep=0pt,minimum size=7pt,opacity=0.2] at (0t) {};
\node[circle,fill=orange,inner sep=0pt,minimum size=7pt,opacity=0.2] at (1t) {};
\end{scope}

\begin{scope}[xshift=5cm]
\primitive{0.5}{0}{0.5}{1}
\qin{i0a}{0}{0.85}{};
\din{i1a}{0}{0.3}{$1$};
\qout{o0a}{1}{0.85}{};
\dout{o1a}{1}{0.3}{$1$};
\node at (0.75,0.3){\scriptsize $S$};\draw[densely dotted] (0.5,0.3) -- (0.65,0.3);\draw[densely dotted] (0.87,0.3) -- (1,0.3);
\node at (0.75,0.85){\scriptsize $\epsilon$};\draw[densely dotted] (0.5,0.85) -- (0.65,0.85);\draw[densely dotted] (0.87,0.85) -- (1,0.85);
\node[circle,fill=blue,inner sep=0pt,minimum size=7pt,opacity=0.2] at (i0a) {};
\node[circle,fill=orange,inner sep=0pt,minimum size=7pt,opacity=0.2] at (i1a) {};
\node[opacity=0.4] at (0.75,-0.25){\scriptsize $\lambda_3$};
\end{scope}

\begin{scope}
\computonComposite{0.15}{-0.3}{2.75}{1.5};

\primitive{0.5}{0}{0.5}{1}
\qin{i0pc}{0}{0.9}{};
\din{i1pc}{0}{0.5}{$1$};
\din{i2pc}{0}{0.1}{$1$};
\node at (0.75,0.3){\scriptsize $\times$};\draw[densely dotted] (0.5,0.5) -- (0.65,0.3);\draw[densely dotted] (0.5,0.1) -- (0.65,0.3);\draw[densely dotted] (0.87,0.3) -- (1,0.3);
\node at (0.75,0.85){\scriptsize $\epsilon$};\draw[densely dotted] (0.5,0.85) -- (0.65,0.85);\draw[densely dotted] (0.85,0.85) -- (1,0.85);

\qmatch{o0pc}{1.5}{0.85}{};\flow{{1,0.9}}{o0pc}{densely dashed}{};\flow{o0pc}{{2,0.85}}{densely dashed}{};
\dmatch{o1pc}{1.5}{0.3}{$1$}{below};\flow{{1,0.3}}{o1pc}{}{};\flow{o1pc}{{2,0.3}}{}{};

\primitive{2}{0}{0.5}{1}
\qout{o0ac}{2.5}{0.85}{};
\dout{o1ac}{2.5}{0.3}{$1$};
\node at (2.25,0.3){\scriptsize $S$};\draw[densely dotted] (2,0.3) -- (2.15,0.3);\draw[densely dotted] (2.37,0.3) -- (2.5,0.3);
\node at (2.25,0.85){\scriptsize $\epsilon$};\draw[densely dotted] (2,0.85) -- (2.15,0.85);\draw[densely dotted] (2.35,0.85) -- (2.5,0.85);

\node[circle,fill=blue,inner sep=0pt,minimum size=7pt,opacity=0.2] at (o0pc) {};
\node[circle,fill=orange,inner sep=0pt,minimum size=7pt,opacity=0.2] at (o1pc) {};
\end{scope}
\end{tikzpicture}
}
{
\begin{tikzpicture}
\begin{scope}
\draw[draw=black,fill=white] (0,0.4) rectangle ++(0.12,0.25);
\node[inner sep=0pt,minimum size=0pt,label=right:{\scriptsize Comp. unit}] at (0.2,0.5) {};
\draw[densely dashed] (2.1,0.5) to node[pos=0.5]{\arrowflow} (2.6,0.5);
\node[inner sep=0pt,minimum size=3pt,label=right:{\scriptsize Control flow}] at (2.5,0.5) {};
\node[draw=black,fill=white,inner sep=0pt,minimum size=3pt,label=right:{\scriptsize Control inport}] at (4.8,0.5) {};
\node[fill=black,inner sep=0pt,minimum size=3pt,label={right:{\scriptsize Control outport}}] at (7.3,0.5) {};
\node[draw,fill=white,fill fraction={black}{0.5},inner sep=0pt,minimum size=3pt,label=right:{\scriptsize Control inoutport}] at (9.8,0.5) {};
\end{scope}

\begin{scope}
\draw (0,0) to node[pos=0.5]{\arrowflow} (0.4,0);
\node[inner sep=0pt,minimum size=3pt,label=right:{\scriptsize Data flow}] at (0.3,0){};
\node[circle,draw=black,fill=white,inner sep=0pt,minimum size=3pt,label={right:{\scriptsize Data inport}}] at (2.1,0) {};
\node[circle,fill=black,inner sep=0pt,minimum size=3pt,label={right:{\scriptsize Data outport}}] at (4.2,0) {};
\node[circle,draw,fill=white,fill fraction={black}{0.5},inner sep=0pt,minimum size=3pt,label=right:{\scriptsize Data inoutport}] at (6.5,0) {};
\draw[draw=black,fill=black!2] (9.1,-0.05) rectangle ++(0.3,0.2);
\node[inner sep=0pt,minimum size=3pt,label=right:{\scriptsize Composite computon}] at (9.35,0){};
\end{scope}
\end{tikzpicture}
}
\vspace{-10pt}
\caption{Examples of partial and total sequencing.}
\label{fig:example-sequencing}
\end{figure}
\vspace{-14pt}
As every computon always possesses control ports, it is always possible to connect at least one computon's control outport with the control inport of another, i.e., every two computons are always sequentiable (see Theorem \ref{th:sequential-always}). Lemmas \ref{prop:interface-sequential} and \ref{prop:interface-sequential-total} specify how interfaces are formed for sequential computons. Lemma \ref{prop:sequencing-mono} specifies that the morphisms induced by the pushout of a sequentiable span are mono. By Theorem \ref{th:sequencing-associative}, both total and partial sequencing are not commutative, and only total sequencing is associative. Theorem \ref{th:sequencing-identity} says that a unit computon serves as the left- and right-identity for both classes of sequencing.

\begin{theorem}\label{th:sequential-always}
If $\lambda_1$ and $\lambda_2$ are computons, there exists a sequentiable span $\rho$ whose pushout is either $\lambda_1\rhd_\rho\lambda_2$ or $\lambda_1\unrhd_\rho\lambda_2$.
\end{theorem}

\begin{lemma}\label{prop:interface-sequential}
Assume ${\rho}$ is a sequentiable span ${\lambda_1\xleftarrow{\alpha_1}\lambda_0\xrightarrow{\alpha_2}\lambda_2}$ of computon morphisms. If ${\lambda_1\xrightarrow{\beta_1}\lambda_3\xleftarrow{\beta_2}\lambda_2}$ is the cospan induced by the pushout of $\rho$, then ${\beta_1(P_1^+)\subseteq P_3^+}$ and ${\beta_2(P_2^-)\subseteq P_3^-}$.
\end{lemma}

\begin{lemma}\label{prop:interface-sequential-total}
Assume that the pushout of ${\rho:=\lambda_1\xleftarrow{\alpha_1}\lambda_0\xrightarrow{\alpha_2}\lambda_2}$ is a total sequential computon $\lambda_3$. If ${\lambda_1\xrightarrow{\beta_1}\lambda_3\xleftarrow{\beta_2}\lambda_2}$ is the cospan induced by the pushout of $\rho$, then ${\beta_1(P_1^+)=P_3^+}$ and ${\beta_2(P_2^-)=P_3^-}$.
\end{lemma}

\begin{lemma}\label{prop:sequencing-mono}
If ${\lambda_1\xrightarrow{\beta_1}\lambda_3\xleftarrow{\beta_2}\lambda_2}$ is the cospan induced by the pushout of a sequentiable span ${\lambda_1\xleftarrow{\alpha_1}\lambda_0\xrightarrow{\alpha_2}\lambda_2}$, then $\beta_1$ and $\beta_2$ are computon monomorphisms.
\end{lemma}

\begin{theorem}\label{th:sequencing-associative}
Total sequencing is associative up to isomorphism, but partial sequencing is not. Both total and partial sequencing are not commutative.
\end{theorem}

\begin{theorem}\label{th:sequencing-identity}
Up to isomorphism, the unit computon is the left- and right-identity for both total and partial sequencing.
\end{theorem}
\vspace{-7pt}
\begin{proof}
It suffices to prove for partial sequencing since the proof of the other is analogous. If we consider that the unit computon $\lambda$ is the left- and right-identity for partial sequencing, we need to show ${\lambda\rhd_{\rho_1}\lambda_1\cong\lambda_1}$ and ${\lambda_1\rhd_{\rho_2}\lambda\cong\lambda_1}$ for some arbitrary computon $\lambda_1$ and sequentiable spans $\rho_1$ and $\rho_2$. We just prove ${\lambda\rhd_{\rho_1}\lambda_1\cong\lambda_1}$ since the other is symmetric. 

By Definition \ref{def:sequentiable-span}, we know that ${\rho_1}$ is a sequentiable span of the form ${\lambda\xleftarrow{\alpha}\lambda_0\xrightarrow{\alpha_1}\lambda_1}$. Since $\lambda_0$ must be a trivial computon, ${U_0\cong\emptyset\cong U}$, ${I_0\cong\emptyset\cong I}$ and ${O_0\cong\emptyset\cong O}$ must hold for the functions ${\alpha_U}$, ${\alpha_I}$ and ${\alpha_O}$ to be well-defined. Given that $\alpha$ must be a monomorphism to satisfy Definition \ref{def:sequentiable-span}, we have ${P_0\cong P}$ in addition. Consequently, ${\lambda_0\cong\lambda}$ holds which directly implies ${\lambda\rhd_{\rho_1}\lambda_1\cong\lambda+_{\lambda_0}\lambda_1\cong\lambda_1}$, as required. 
\end{proof}
\vspace{-10pt}
\subsection{Parallel Computons}
\label{sec:operators-parallel}

In any computon category, it is possible to form composites to encapsulate control flow for the asynchronous invocation of two computons, as formalised in Definition \ref{def:async-computon}.

\begin{definition}[Async]\label{def:async-computon}
An async computon ${\lambda_1+\lambda_2}$ is the coproduct of $\lambda_1$ and $\lambda_2$.
\end{definition}

Given that a computon category has all coproducts \cite{arellanescompositional2026}, an async can always be formed (see Theorem \ref{th:async-exists}). By Theorems \ref{th:async-commutativity} and \ref{th:async-identity}, asynchronous parallelising (i.e., the operation to form an async) is associative and commutative, but has no identity.

\begin{theorem}\label{th:async-exists}
An async computon ${\lambda_1+\lambda_2}$ can always be constructed in a category of computons.
\end{theorem}

\begin{theorem}\label{th:async-commutativity}
Asynchronous parallelising is associative and commutative up to isomorphism.
\end{theorem}

\begin{theorem}\label{th:async-identity}
Asynchronous parallelising does not satisfy the identity law.
\end{theorem}
\vspace{-7pt}
\begin{proof}
Suppose for contradiction that there is a computon $\lambda$ serving as the left- and right-identity for asynchronous parallelising so that ${\lambda+\lambda_1\cong\lambda_1\cong\lambda_1+\lambda}$ for any computon $\lambda_1$. As ${\lambda+\lambda_1}$ is the coproduct of $\lambda$ and $\lambda_1$ (see Definition \ref{def:async-computon}), Definition \ref{def:coproduct} states ${P\sqcup P_1\cong P_1}$, which is only true if $P=\emptyset$. But Definition \ref{def:computon} says that the set of ports cannot be empty so $\lambda$ cannot be a left-identity. Disproving the existence of a right-identity follows analogously. 
\end{proof}
\vspace{-5pt}
In \cite{arellanescompositional2026}, an operator for synchronous parallelising was described, relying on a special class of computons referred to as joins, which are primitive computons with exactly two control inports and one control outport. In the present work, we generalise such a notion in the form of a so-called \emph{glue}.

\begin{definition}[Glue]\label{def:glue}
A glue is a primitive computon $\lambda$ with ${c(p)=0}$ for all ${p \in P}$.
\end{definition}

As a join is just a special glue, glues can be used to build \emph{sync composites} to synchronise the execution of multiple computons (see Definition \ref{def:sync-computon}). By Theorem \ref{th:always-sync}, such structures can always be formed.

\begin{definition}[Sync]\label{def:sync-computon}
Given computons $\lambda_1$ and $\lambda_2$, a sync computon is the composite ${(\lambda_1+\lambda_2)\square\lambda_3}$ where $\lambda_3$ is a glue, ${\square\in\{\rhd_{\rho},\unrhd_{\rho}\}}$ and $\rho$ is a sequentiable span that identifies all the control outports of ${\lambda_1+\lambda_2}$ with all the inports of $\lambda_3$.
\end{definition}

\begin{remark}
Although every chosen glue computon $\lambda_3$ is required to satisfy ${P_3^+\cong C_1^-\sqcup C_2^-}$ with $C_1^-:=\{p\in P_1^-\mid c_1(p)=0\}$ and ${C_2^-:=\{p\in P_2^-\mid c_2(p)=0\}}$, there are no specific requirements for $P_3^-$. So, the choice of $\lambda_3$ is not uniquely defined. Similarly, the choice of $\rho$ is not unique because the pushout can yield a total or a partial sequential computon. Therefore, a sync computon is not uniquely defined.
\end{remark}

\begin{theorem}\label{th:always-sync}
A sync computon can always be constructed in a category of computons.
\end{theorem}
\begin{proof}
Given objects $\lambda_1$ and $\lambda_2$ in a computon category over a set $\underline{n}$ of types and a set $B$ of computing devices, we define ${C_j^-:=\{p\in P_j^-\mid c_j(p)=0\}}$ for ${j\in\{1,2\}}$ and the coproduct ${(C_1^-\sqcup C_2^-)\sqcup (C_1^-\sqcup C_2^-)}$ with canonical injections ${\iota_1,\iota_2\colon C_1^-\sqcup C_2^-\to (C_1^-\sqcup C_2^-)\sqcup (C_1^-\sqcup C_2^-)}$. With this, we construct a glue $\lambda_3$: ${U_3=\{u\}}$, ${P_3\cong (C_1^-\sqcup C_2^-)\sqcup (C_1^-\sqcup C_2^-)}$, ${I_3\cong C_1^-\sqcup C_2^-}$ and ${O_3\cong C_1^-\sqcup C_2^-}$. If $\phi\colon P_3\to (C_1^-\sqcup C_2^-)\sqcup (C_1^-\sqcup C_2^-)$, ${\gamma\colon I_3\to C_1^-\sqcup C_2^-}$ and ${\zeta\colon O_3\to C_1^-\sqcup C_2^-}$ are the isomorphisms, we have: $s_3=\phi^{-1}\circ\iota_1\circ\gamma$, $t_3=\phi^{-1}\circ\iota_2\circ\zeta$, ${\sigma_3(o)=u=\tau_3(i)}$, ${c_3(p)=0}$, ${f_3(o)=\epsilon}$ and ${r_3=\zeta^{-1}\circ\gamma}$ for all ${i\in I_3}$, ${o\in O_3}$ and ${p\in P_3}$.

Since $\phi^{-1}$, $\gamma$ and $\zeta$ are isomorphisms and $\iota_1$ and $\iota_2$ are necessarily injective, $s_3$ and $t_3$ are injective too. The functions $\sigma_3$ and $\tau_3$ are surjective because they map flows to the unique ${u\in U_3}$. The function $r_3$ is surjective because $\gamma$ and $\zeta^{-1}$ are. Since ${|U_3|=1}$ and ${P_3\cong (C_1^-\sqcup C_2^-)\sqcup (C_1^-\sqcup C_2^-)\cong I_3 \sqcup O_3}$, $\lambda_3$ is primitive as per Definition \ref{def:primitive-computon}. It also is a glue because all its ports are zero-coloured through $c_3$. 

Now, consider the trivial computon $\lambda_4$ given by ${P_4\cong I_3}$ with ${c_4(p)=0}$ for all ${p\in P_4}$. Assuming ${\psi\colon P_4\to I_3}$ is the corresponding isomorphism, we construct the computon morphism ${\alpha_1\colon\lambda_4\to\lambda_1+\lambda_2}$ by taking ${\gamma\circ\psi}$ as the $P$-component, and the identity functions on $\underline{n}$ and $B$ as the $\underline{n}$- and $B$-components, respectively. Here, the function ${\gamma\circ\psi\colon P_4\to C_1^-\sqcup C_2^-}$ is well-defined because ${C_1^-\sqcup C_2^-\subseteq P_1\sqcup P_2}$. Having the inclusions ${\vec{i}(\alpha_1)\cup\vec{o}(\alpha_1)\subseteq P_4=P_4^+\cap P_4^-\subseteq P_4^+\cup P_4^-}$, $\alpha_1$ must be a valid computon morphism as per Definition \ref{def:morphism}. In fact, it is a monomorphism because all its components, including ${\gamma\circ\psi}$, are injective.

Now, construct ${\alpha_2\colon\lambda_4\to\lambda_3}$ by taking ${s_3\circ\psi}$ and the identities on ${\underline{n}}$ and ${B}$ as the corresponding $P$-, $\underline{n}$- and $B$-components. Checking that ${\alpha_2}$ satisfies Definition \ref{def:morphism} is analogous to $\alpha_1$, so $\alpha_2$ is also a computon morphism. Given that all the ${\alpha_2}$-components are injective, including ${s_3\circ\psi}$, ${\alpha_2}$ is also a monomorphism.

As Definition \ref{def:interface} entails ${\alpha_2(p)=s_3(\psi(p))\notin P_3^-}$ for all ${p\in P_4}$, we simply apply Proposition \ref{prop:primitive-interface} to deduce ${\alpha_2(p)\in P_3^+}$, i.e., ${\alpha_2(P_4)\subseteq P_3^+}$. Having the fact ${(\gamma\circ\psi)(P_4)\subseteq C_1^-\sqcup C_2^-\subseteq P_1^-\sqcup P_2^-}$ in addition and considering that $\alpha_1$ and $\alpha_2$ are mono, the span ${\rho:=(\lambda_1+\lambda_2)\xleftarrow{\alpha_1}\lambda_4\xrightarrow{\alpha_2}\lambda_3}$ must be sequentiable (see Definition \ref{def:sequentiable-span}). Therefore, ${(\lambda_1+\lambda_2)\rhd_{\rho}\lambda_3}$ or ${(\lambda_1+\lambda_2)\unrhd_{\rho}\lambda_3}$ can be formed. 
\end{proof}

Basically, the proof of Theorem \ref{th:always-sync} yields a 2-stage construction to form a sync composite out of computons $\lambda_1$ and $\lambda_2$. The idea is to first define ${\lambda_1+\lambda_2}$ (for parallel execution) and then use sequencing to connect all the control outports of ${\lambda_1+\lambda_2}$ with all the inports of a glue (which waits for the async's termination). Evidently, the apex of the sequentiable span must be a trivial computon with control ports only. 

As we have been discussing composition semantics only, it might not seem entirely obvious how a sync behaves. To elucidate this, let us consider an example.

\begin{example}
Figs. \ref{fig:example-parallelising}a and \ref{fig:example-parallelising}b use the same primitive computons as Fig. \ref{fig:example-sequencing} to respectively form ${\lambda_1+\lambda_2}$ and ${(\lambda_1+\lambda_2)\rhd_{\rho}\lambda}$, in order to perform binary product and binary addition in parallel. The only difference is that \ref{fig:example-parallelising}b awaits termination via the glue $\lambda$. Although the sequentiable span $\rho$ is not shown due to space constraints, it should be clear that it meets Definition \ref{def:sync-computon} as it identifies all the control outports of ${\lambda_1+\lambda_2}$ with all the inports of $\lambda$. In this case, the identification corresponds to partial sequencing because data ports are omitted. Consequently, there is an effect in which data ports traverse several composite layers.
\end{example}

\begin{figure}[H]
\centering
\subcaptionbox{${\lambda_1+\lambda_2}$.}[0.34\textwidth]
{
\begin{tikzpicture}[scale=0.76]
\computonComposite{0.1}{-0.3}{1.3}{2.8};
\primitive{0.5}{1.2}{0.5}{1}
\din{i1p}{0}{2.05}{$1$};
\din{i2p}{0}{1.7}{$1$};
\qin{i0p}{0}{1.35}{};
\dout{o1p}{1}{1.9}{$1$};
\qout{o0p}{1}{1.35}{};
\node at (0.75,1.9){\scriptsize $\times$};\draw[densely dotted] (0.5,2.05) -- (0.65,1.9);\draw[densely dotted] (0.5,1.7) -- (0.65,1.85);\draw[densely dotted] (0.87,1.9) -- (1,1.9);
\node at (0.75,1.35){\scriptsize $\epsilon$};\draw[densely dotted] (0.5,1.35) -- (0.65,1.35);\draw[densely dotted] (0.85,1.35) -- (1,1.35);
\node[opacity=0.4] at (0.75,2.35){\scriptsize $\lambda_1$};

\primitive{0.5}{0}{0.5}{1}
\qin{i0a}{0}{0.85}{};
\din{i1a}{0}{0.5}{$1$};
\din{i2a}{0}{0.1}{$2$};
\qout{o0a}{1}{0.85}{};
\dout{o1a}{1}{0.3}{$2$};
\node at (0.75,0.3){\scriptsize $+$};\draw[densely dotted] (0.5,0.5) -- (0.65,0.3);\draw[densely dotted] (0.5,0.1) -- (0.65,0.3);\draw[densely dotted] (0.87,0.3) -- (1,0.3);
\node at (0.75,0.85){\scriptsize $\epsilon$};\draw[densely dotted] (0.5,0.85) -- (0.65,0.85);\draw[densely dotted] (0.85,0.85) -- (1,0.85);
\node[opacity=0.4] at (0.75,-0.15){\scriptsize $\lambda_2$};
\end{tikzpicture}
}
\subcaptionbox{${(\lambda_1+\lambda_2)\rhd_{\rho}\lambda}$.}[0.5\textwidth]
{
\begin{tikzpicture}[scale=0.76]
\computonComposite{1.65}{-0.4}{2.7}{3};
\computonComposite{1.85}{-0.3}{0.95}{2.8};

\begin{scope}[xshift=1.5cm]
\primitive{0.5}{1.2}{0.5}{1}
\qout{o0p}{1}{1.4}{};
\node at (0.75,1.9){\scriptsize $\times$};\draw[densely dotted] (0.5,2.05) -- (0.65,1.85);\draw[densely dotted] (0.5,1.65) -- (0.65,1.85);\draw[densely dotted] (0.87,1.85) -- (1,1.85);
\node at (0.75,1.4){\scriptsize $\epsilon$};\draw[densely dotted] (0.5,1.4) -- (0.65,1.4);\draw[densely dotted] (0.85,1.4) -- (1,1.4);
\node[opacity=0.4] at (0.75,2.35){\scriptsize $\lambda_1$};

\primitive{0.5}{0}{0.5}{1}
\qout{o0a}{1}{0.85}{};
\node at (0.75,0.3){\scriptsize $+$};\draw[densely dotted] (0.5,0.5) -- (0.65,0.3);\draw[densely dotted] (0.5,0.1) -- (0.65,0.3);\draw[densely dotted] (0.87,0.3) -- (1,0.3);
\node at (0.75,0.85){\scriptsize $\epsilon$};\draw[densely dotted] (0.5,0.85) -- (0.65,0.85);\draw[densely dotted] (0.85,0.85) -- (1,0.85);
\node[opacity=0.4] at (0.75,-0.15){\scriptsize $\lambda_2$};

\dinplain{i1p}{0}{2.05}{$1$}{left};\flow{i1p}{{0.5,2.05}}{}{};
\dinplain{i2p}{0}{1.7}{$1$}{left};\flow{i2p}{{0.5,1.7}}{}{};
\qinplain{i0p}{0}{1.4}{};\flow{i0p}{{0.5,1.4}}{densely dashed}{};
\qinplain{i0a}{0}{0.85}{};\flow{i0a}{{0.5,0.85}}{densely dashed}{};
\dinplain{i1a}{0}{0.5}{$1$}{left};\flow{i1a}{{0.5,0.5}}{}{};
\dinplain{i2a}{0}{0.1}{$2$}{left};\flow{i2a}{{0.5,0.1}}{}{};
\end{scope}

\begin{scope}[xshift=3cm]
\qmatch{i0j}{0}{1.4}{};\flow{i0j}{{0.5,1.4}}{densely dashed}{};
\qmatch{i1j}{0}{0.85}{};\flow{i1j}{{0.5,0.85}}{densely dashed}{};

\primitive{0.5}{0.6}{0.5}{1}
\qout{o0j}{1}{1.1}{};

\node at (0.75,1.1){\scriptsize $\epsilon$};\draw[densely dotted] (0.5,0.85) -- (0.7,1.05);\draw[densely dotted] (0.5,1.4) -- (0.7,1.15);\draw[densely dotted] (0.81,1.1) -- (1,1.1);
\doutplain{o1p}{1}{1.85}{$1$}{right};\flow{{-0.5,1.85}}{o1p}{}{};
\doutplain{o1a}{1}{0.3}{$2$}{right};\flow{{-0.5,0.3}}{o1a}{}{};
\node[opacity=0.4] at (0.75,1.75){\scriptsize $\lambda$};
\end{scope}
\end{tikzpicture}
}
{
\begin{tikzpicture}
\begin{scope}
\draw[draw=black,fill=white] (0,0.4) rectangle ++(0.12,0.25);
\node[inner sep=0pt,minimum size=0pt,label=right:{\scriptsize Comp. unit}] at (0.2,0.5) {};
\draw[densely dashed] (2.1,0.5) to node[pos=0.5]{\arrowflow} (2.6,0.5);
\node[inner sep=0pt,minimum size=3pt,label=right:{\scriptsize Control flow}] at (2.5,0.5) {};
\node[draw=black,fill=white,inner sep=0pt,minimum size=3pt,label=right:{\scriptsize Control inport}] at (4.8,0.5) {};
\node[fill=black,inner sep=0pt,minimum size=3pt,label={right:{\scriptsize Control outport}}] at (7.3,0.5) {};
\node[draw,fill=white,fill fraction={black}{0.5},inner sep=0pt,minimum size=3pt,label=right:{\scriptsize Control inoutport}] at (9.8,0.5) {};
\end{scope}
\begin{scope}
\draw (0,0) to node[pos=0.5]{\arrowflow} (0.4,0);
\node[inner sep=0pt,minimum size=3pt,label=right:{\scriptsize Data flow}] at (0.3,0){};
\node[circle,draw=black,fill=white,inner sep=0pt,minimum size=3pt,label={right:{\scriptsize Data inport}}] at (2.1,0) {};
\node[circle,fill=black,inner sep=0pt,minimum size=3pt,label={right:{\scriptsize Data outport}}] at (4.2,0) {};
\node[circle,draw,fill=white,fill fraction={black}{0.5},inner sep=0pt,minimum size=3pt,label=right:{\scriptsize Data inoutport}] at (6.5,0) {};
\draw[draw=black,fill=black!2] (9.1,-0.05) rectangle ++(0.3,0.2);
\node[inner sep=0pt,minimum size=3pt,label=right:{\scriptsize Composite computon}] at (9.35,0){};
\end{scope}
\end{tikzpicture}
}
\vspace{-5pt}
\caption{Examples of asynchronous and synchronous parallelising.}
\label{fig:example-parallelising}
\end{figure}

\subsection{Branching Computons}

A branching computon embodies an essential structure for non-deterministic decision-making, allowing to choose a computon out of two alternatives. Such a class of composites can be either open or closed. The former is simply the pushout of a span of in-markers of the operands (see Definition \ref{def:branched-open}), whereas the latter is the colimit of a so-called \emph{b-diagram} which basically is a span of in-markers together with a span of out-markers (see Definition \ref{def:branched-closed}). In both cases, the in-markers are embedded into every operand inport so a branching structure can only be formed when inports fully match. In the case of open branching, the outports of the operands remain untouched because there are no out-markers; thus, contrasting with the full outport identification enforced by closed branching.

\begin{definition}[Open Branching Computon]\label{def:branched-open}
An open branching computon ${\lambda_2?_{\rho}\lambda_3}$ is the pushout of a span ${\rho:=\lambda_2\xleftarrow{\lambda_2^+}\lambda_0\xrightarrow{\lambda_3^+}\lambda_3}$.
\end{definition}

\begin{definition}[Closed Branching Computon]\label{def:branched-closed}
A b-diagram $\rho$ is a pair of spans, ${\lambda_2\xleftarrow{\lambda_2^+}\lambda_0\xrightarrow{\lambda_3^+}\lambda_3}$ and ${\lambda_2\xleftarrow{\lambda_2^-}\lambda_1\xrightarrow{\lambda_3^-}\lambda_3}$, where $\lambda_2$ and $\lambda_3$ are connected computons. A closed branching computon $\lambda_2??_{\rho}\lambda_3$ is the colimit of $\rho$, computed as ${\lambda_2+_{\lambda_0+\lambda_1}\lambda_3}$.
\end{definition}

\begin{example}
To elucidate Definitions \ref{def:branched-open} and \ref{def:branched-closed}, suppose that in addition to the computon $\lambda_3$ from Fig. \ref{fig:example-sequencing}, we have primitives $\lambda_4$ and $\lambda_5$ for computing the predecessor of a natural number ($P$) and the factorial function ($!$), respectively. With them, we can form the branching computons displayed in Fig. \ref{fig:example-branching}.

\begin{figure}[H]
\subcaptionbox{Constructing ${\lambda_3?_{\rho_1}\lambda_4}$.}[0.4\textwidth]
{
\begin{tikzpicture}
\draw[thick,-{Stealth[inset=0pt, length=3pt, angle'=90]},opacity=0.4] (0,1.7) to node[above]{\scriptsize $\lambda_3^+$} (0.6,1.9);
\draw[thick,-{Stealth[inset=0pt, length=3pt, angle'=90]},opacity=0.4] (0,0.8) to node[below]{\scriptsize $\lambda_4^+$} (0.6,0.5);
\draw[thick,-{Stealth[inset=0pt, length=3pt, angle'=90]},opacity=0.4] (2.8,2) -- (3.3,1.8);
\draw[thick,-{Stealth[inset=0pt, length=3pt, angle'=90]},opacity=0.4] (2.8,0.6) -- (3.3,0.8);

\begin{scope}[yshift=1cm]
\qmatch{0t}{0}{0.4}{};
\dmatch{1t}{0}{0}{$1$}{left};
\node[circle,fill=blue,inner sep=0pt,minimum size=7pt,opacity=0.2] at (0t) {};
\node[circle,fill=orange,inner sep=0pt,minimum size=7pt,opacity=0.2] at (1t) {};
\end{scope}

\begin{scope}[xshift=0.7cm,yshift=0.1cm]
\primitive{0.8}{1.2}{0.5}{1}
\node at (1.05,1.5){\scriptsize $S$};\draw[densely dotted] (0.8,1.5) -- (0.95,1.5);\draw[densely dotted] (1.17,1.5) -- (1.3,1.5);
\node at (1.05,2.05){\scriptsize $\epsilon$};\draw[densely dotted] (0.8,2.05) -- (0.95,2.05);\draw[densely dotted] (1.15,2.05) -- (1.3,2.05);
\qin{i0s}{0.3}{2.05}{}
\din{i1s}{0.3}{1.5}{$1$};
\qout{o0s}{1.3}{2.05}{};
\dout{o1s}{1.3}{1.5}{$1$};
\node[opacity=0.4] at (1.05,2.35){\scriptsize $\lambda_3$};
\node[circle,fill=blue,inner sep=0pt,minimum size=7pt,opacity=0.2] at (i0s) {};
\node[circle,fill=orange,inner sep=0pt,minimum size=7pt,opacity=0.2] at (i1s) {};
\end{scope}

\begin{scope}[xshift=0.7cm]
\primitive{0.8}{0}{0.5}{1}
\node at (1.05,0.3){\scriptsize $P$};\draw[densely dotted] (0.8,0.3) -- (0.95,0.3);\draw[densely dotted] (1.17,0.3) -- (1.3,0.3);
\node at (1.05,0.85){\scriptsize $\epsilon$};\draw[densely dotted] (0.8,0.85) -- (0.95,0.85);\draw[densely dotted] (1.15,0.85) -- (1.3,0.85);
\qin{i0f}{0.3}{0.85}{}
\din{i1f}{0.3}{0.3}{$1$};
\qout{o0f}{1.3}{0.85}{};
\dout{o1f}{1.3}{0.3}{$2$};
\node[opacity=0.4] at (1.05,-0.15){\scriptsize $\lambda_4$};
\node[circle,fill=blue,inner sep=0pt,minimum size=7pt,opacity=0.2] at (i0f) {};
\node[circle,fill=orange,inner sep=0pt,minimum size=7pt,opacity=0.2] at (i1f) {};
\end{scope}

\begin{scope}[xshift=3.3cm]
\computonComposite{0.15}{-0.2}{1.2}{2.6};
\primitive{0.5}{1.2}{0.5}{1}
\node at (0.75,1.5){\scriptsize $S$};\draw[densely dotted] (0.5,1.5) -- (0.65,1.5);\draw[densely dotted] (0.87,1.5) -- (1,1.5);
\node at (0.75,2.05){\scriptsize $\epsilon$};\draw[densely dotted] (0.5,2.05) -- (0.65,2.05);\draw[densely dotted] (0.85,2.05) -- (1,2.05);
\qout{o0s}{1}{2.05}{};
\dout{o1s}{1}{1.5}{$1$};

\primitive{0.5}{0}{0.5}{1}
\node at (0.75,0.3){\scriptsize $P$};\draw[densely dotted] (0.5,0.3) -- (0.65,0.3);\draw[densely dotted] (0.87,0.3) -- (1,0.3);
\node at (0.75,0.85){\scriptsize $\epsilon$};\draw[densely dotted] (0.5,0.85) -- (0.65,0.85);\draw[densely dotted] (0.85,0.85) -- (1,0.85);
\qout{o0f}{1}{0.85}{};
\dout{o1f}{1}{0.3}{$2$};

\qinplain{i0}{0}{1.4}{};\flowdiag{i0}{{0.5,2.05}}{dashed}{}{pos=0.5,rotate=45};\flowdiag{i0}{{0.5,0.85}}{dashed}{}{pos=0.7,rotate=315};	
\dinplain{i1}{0}{1}{$1$}{left};\flowdiag{i1}{{0.5,1.5}}{}{}{pos=0.7,rotate=45};\flowdiag{i1}{{0.5,0.3}}{}{}{pos=0.5,rotate=315};

\node[circle,fill=blue,inner sep=0pt,minimum size=7pt,opacity=0.2] at (i0) {};
\node[circle,fill=orange,inner sep=0pt,minimum size=7pt,opacity=0.2] at (i1) {};
\end{scope}
\end{tikzpicture}
}
\subcaptionbox{Constructing ${\lambda_3??_{\rho_2}\lambda_5}$.}[0.5\textwidth]
{
\begin{tikzpicture}
\draw[thick,-{Stealth[inset=0pt, length=3pt, angle'=90]},opacity=0.4] (0.2,2.2) to node[above]{\scriptsize $\lambda_3^+$} (0.6,2.2);
\draw[thick,-{Stealth[inset=0pt, length=3pt, angle'=90]},opacity=0.4] (0.2,2) to node[left]{\scriptsize $\lambda_5^+$} (1.1,1.1);
\draw[thick,-{Stealth[inset=0pt, length=3pt, angle'=90]},opacity=0.4] (0.2,0.6) to node[left]{\scriptsize $\lambda_3^-$} (1,1.5);
\draw[thick,-{Stealth[inset=0pt, length=3pt, angle'=90]},opacity=0.4] (0.2,0.4) to node[below]{\scriptsize $\lambda_5^-$} (0.6,0.4);
\draw[decorate, decoration={snake, amplitude=0.8mm, segment length=4mm}, ->,opacity=0.4] (2,1.2) to node[above]{\scalebox{0.9}{\tiny ${\lambda_3+_{\lambda_0+\lambda_1}\lambda_5}$}} (3.7,1.2);

\begin{scope}[yshift=2cm]
\node[opacity=0.4] at (-0.4,0.3){\scriptsize $\lambda_0$};
\qmatch{0t}{0}{0.4}{};
\dmatch{1t}{0}{0}{$1$}{left};
\node[circle,fill=blue,inner sep=0pt,minimum size=7pt,opacity=0.2] at (0t) {};
\node[circle,fill=orange,inner sep=0pt,minimum size=7pt,opacity=0.2] at (1t) {};
\end{scope}
\begin{scope}[yshift=0.3cm]
\node[opacity=0.4] at (-0.4,0.3){\scriptsize $\lambda_1$};
\qmatch{0t1}{0}{0.4}{};
\dmatch{1t1}{0}{0}{$1$}{left};
\node[circle,fill=green,inner sep=0pt,minimum size=7pt,opacity=0.2] at (0t1) {};
\node[circle,fill=black,inner sep=0pt,minimum size=7pt,opacity=0.2] at (1t1) {};
\end{scope}

\begin{scope}[xshift=0.65cm,yshift=0.3cm]
\primitive{0.8}{1.2}{0.5}{1}
\node at (1.05,1.5){\scriptsize $S$};\draw[densely dotted] (0.8,1.5) -- (0.95,1.5);\draw[densely dotted] (1.17,1.5) -- (1.3,1.5);
\node at (1.05,2.05){\scriptsize $\epsilon$};\draw[densely dotted] (0.8,2.05) -- (0.95,2.05);\draw[densely dotted] (1.15,2.05) -- (1.3,2.05);
\qin{i0s}{0.3}{2.05}{}
\din{i1s}{0.3}{1.5}{$1$};
\qout{o0s}{1.3}{2.05}{};
\dout{o1s}{1.3}{1.5}{$1$};
\node[opacity=0.4] at (1.05,2.35){\scriptsize $\lambda_3$};
\node[circle,fill=blue,inner sep=0pt,minimum size=7pt,opacity=0.2] at (i0s) {};
\node[circle,fill=orange,inner sep=0pt,minimum size=7pt,opacity=0.2] at (i1s) {};
\node[circle,fill=green,inner sep=0pt,minimum size=7pt,opacity=0.2] at (o0s) {};
\node[circle,fill=black,inner sep=0pt,minimum size=7pt,opacity=0.2] at (o1s) {};
\end{scope}

\begin{scope}[xshift=0.65cm]
\primitive{0.8}{0}{0.5}{1}
\node at (1.05,0.3){\scriptsize $!$};\draw[densely dotted] (0.8,0.3) -- (0.95,0.3);\draw[densely dotted] (1.17,0.3) -- (1.3,0.3);
\node at (1.05,0.85){\scriptsize $\epsilon$};\draw[densely dotted] (0.8,0.85) -- (0.95,0.85);\draw[densely dotted] (1.15,0.85) -- (1.3,0.85);
\qin{i0f}{0.3}{0.85}{}
\din{i1f}{0.3}{0.3}{$1$};
\qout{o0f}{1.3}{0.85}{};
\dout{o1f}{1.3}{0.3}{$1$};
\node[opacity=0.4] at (1.05,-0.15){\scriptsize $\lambda_5$};
\node[circle,fill=blue,inner sep=0pt,minimum size=7pt,opacity=0.2] at (i0f) {};
\node[circle,fill=orange,inner sep=0pt,minimum size=7pt,opacity=0.2] at (i1f) {};
\node[circle,fill=green,inner sep=0pt,minimum size=7pt,opacity=0.2] at (o0f) {};
\node[circle,fill=black,inner sep=0pt,minimum size=7pt,opacity=0.2] at (o1f) {};
\end{scope}

\begin{scope}[xshift=3.9cm]
\computonComposite{0.2}{-0.2}{1.1}{2.6};

\primitive{0.5}{1.2}{0.5}{1}
\node at (0.75,1.5){\scriptsize $S$};\draw[densely dotted] (0.5,1.5) -- (0.65,1.5);\draw[densely dotted] (0.87,1.5) -- (1,1.5);
\node at (0.75,2.05){\scriptsize $\epsilon$};\draw[densely dotted] (0.5,2.05) -- (0.65,2.05);\draw[densely dotted] (0.85,2.05) -- (1,2.05);

\primitive{0.5}{0}{0.5}{1}
\node at (0.75,0.3){\scriptsize $!$};\draw[densely dotted] (0.5,0.3) -- (0.65,0.3);\draw[densely dotted] (0.87,0.3) -- (1,0.3);
\node at (0.75,0.85){\scriptsize $\epsilon$};\draw[densely dotted] (0.5,0.85) -- (0.65,0.85);\draw[densely dotted] (0.85,0.85) -- (1,0.85);

\qinplain{i0}{0}{1.4}{};\flowdiag{i0}{{0.5,2.05}}{dashed}{}{pos=0.5,rotate=45};\flowdiag{i0}{{0.5,0.85}}{dashed}{}{pos=0.7,rotate=315};	
\dinplain{i1}{0}{1}{$1$}{left};\flowdiag{i1}{{0.5,1.5}}{}{}{pos=0.7,rotate=45};\flowdiag{i1}{{0.5,0.3}}{}{}{pos=0.5,rotate=315};
\qoutplain{o0}{1}{1.4}{};\flowdiag{{1,2.05}}{o0}{dashed}{}{pos=0.5,rotate=315};\flowdiag{{1,0.85}}{o0}{dashed}{}{pos=0.3,rotate=45};	
\doutplain{o1}{1}{1}{$1$}{right};\flowdiag{{1,1.5}}{o1}{}{}{pos=0.3,rotate=315};\flowdiag{{1,0.3}}{o1}{}{}{pos=0.6,rotate=45};

\node[circle,fill=blue,inner sep=0pt,minimum size=7pt,opacity=0.2] at (i0) {};
\node[circle,fill=orange,inner sep=0pt,minimum size=7pt,opacity=0.2] at (i1) {};
\node[circle,fill=green,inner sep=0pt,minimum size=7pt,opacity=0.2] at (o0) {};
\node[circle,fill=black,inner sep=0pt,minimum size=7pt,opacity=0.2] at (o1) {};	
\end{scope}
\end{tikzpicture}
}
{
\begin{center}
\begin{tikzpicture}
\begin{scope}
\draw[draw=black,fill=white] (0,0.4) rectangle ++(0.12,0.25);
\node[inner sep=0pt,minimum size=0pt,label=right:{\scriptsize Comp. unit}] at (0.2,0.5) {};
\draw[densely dashed] (2.1,0.5) to node[pos=0.5]{\arrowflow} (2.6,0.5);
\node[inner sep=0pt,minimum size=3pt,label=right:{\scriptsize Control flow}] at (2.5,0.5) {};
\node[draw=black,fill=white,inner sep=0pt,minimum size=3pt,label=right:{\scriptsize Control inport}] at (4.8,0.5) {};
\node[fill=black,inner sep=0pt,minimum size=3pt,label={right:{\scriptsize Control outport}}] at (7.3,0.5) {};
\node[draw,fill=white,fill fraction={black}{0.5},inner sep=0pt,minimum size=3pt,label=right:{\scriptsize Control inoutport}] at (9.8,0.5) {};
\end{scope}

\begin{scope}
\draw (0,0) to node[pos=0.5]{\arrowflow} (0.4,0);
\node[inner sep=0pt,minimum size=3pt,label=right:{\scriptsize Data flow}] at (0.3,0){};
\node[circle,draw=black,fill=white,inner sep=0pt,minimum size=3pt,label={right:{\scriptsize Data inport}}] at (2.1,0) {};
\node[circle,fill=black,inner sep=0pt,minimum size=3pt,label={right:{\scriptsize Data outport}}] at (4.2,0) {};
\node[circle,draw,fill=white,fill fraction={black}{0.5},inner sep=0pt,minimum size=3pt,label=right:{\scriptsize Data inoutport}] at (6.5,0) {};
\draw[draw=black,fill=black!2] (9.1,-0.05) rectangle ++(0.3,0.2);
\node[inner sep=0pt,minimum size=3pt,label=right:{\scriptsize Composite computon}] at (9.35,0){};
\end{scope}
\end{tikzpicture}
\end{center}
}
\vspace{-5pt}
\caption{Open and closed branching. The wavy arrow is not a morphism, but it expresses the operation ${\lambda_3+_{\lambda_0+\lambda_1}\lambda_5}$ to form the closed branching computon ${\lambda_3??_{\rho_2}\lambda_5}$ from the b-diagram $\rho_2$ shown on the left of (b). The open branching computon ${\lambda_3?_{\rho_1}\lambda_4}$ is the pushout of the span $\rho_1$ of in-markers ${\lambda_3^+}$ and ${\lambda_4^+}$.}
\label{fig:example-branching}
\end{figure}

Fig. \ref{fig:example-branching}a shows an example of an open branching composite able to choose between the successor and predecessor primitives in a non-deterministic manner, whereas Fig. \ref{fig:example-branching}b displays a closed branching composite for choosing either the successor or the factorial function. Fig. \ref{fig:example-branching}a particularly unfolds the flexibility unleashed by the operator described in Definition \ref{def:branched-open}, i.e., the operand outports do not have to fully match as in the case of closed branching. 
\end{example}

By Theorems \ref{th:open} and \ref{th:closed-assoc-comm}, both operators satisfy associativity and commutativity, but only open branching has left- and right-identities (see Theorems \ref{th:open2} and \ref{th:closed-identity}).

\begin{theorem}\label{th:open}
Open branching is associative and commutative up to isomorphism.
\end{theorem}

\begin{theorem}\label{th:open2}
If ${\rho}$ is the span ${\lambda_2\xleftarrow{\lambda_2^+}\lambda_0\xrightarrow{\lambda_3^+}\lambda_3}$ of in-markers with ${\lambda_0\cong\lambda_2}$, then ${\lambda_2?_{\rho}\lambda_3\cong\lambda_2+_{\lambda_0}\lambda_3\cong\lambda_3}$.
\end{theorem}
\begin{proof}
As $\lambda_2$ must evidently be a trivial computon because ${\lambda_0\cong\lambda_2}$, we use Definition \ref{def:pushout} to deduce that ${\lambda_2?_{\rho}\lambda_3}$ has ${|U_3|}$ units, ${|P_2|+|P_3|-|P_0|}$ ports, ${|I_3|}$ inflows and ${|O_3|}$ outflows. Particularly, ${\lambda_2?_{\rho}\lambda_3}$ has ${|P_3|}$ ports because ${|P_0|=|P_2|}$ by ${\lambda_0\cong\lambda_2}$. As the sets of types and devices is the same for every computon in the same category, we conclude ${\lambda_2?_{\rho}\lambda_3\cong\lambda_2+_{\lambda_0}\lambda_3\cong\lambda_3}$.
\end{proof}

\begin{theorem}[\cite{arellanescompositional2026}]\label{th:closed-assoc-comm}
Closed branching is associative and commutative up to isomorphism.
\end{theorem}

\begin{theorem}\label{th:closed-identity}
Closed branching does not satisfy the identity law.
\end{theorem}
\begin{proof}
Suppose for contradiction that a computon $\lambda_2$ is the left-identity of closed branching so $\lambda_2??_{\rho}\lambda_3\cong\lambda_3$ for some computon $\lambda_3$ and some b-diagram $\rho$ formed by spans ${\lambda_2\xleftarrow{\lambda_2^+}\lambda_0\xrightarrow{\lambda_3^+}\lambda_3}$ and ${\lambda_2\xleftarrow{\lambda_2^-}\lambda_1\xrightarrow{\lambda_3^-}\lambda_3}$. As $\lambda_2$ and $\lambda_3$ must be connected as per Definition \ref{def:branched-closed}, Proposition \ref{prop:connected-units} says that units ${u_2\in U_2}$ and ${u_3\in U_3}$ must exist. If we assume that ${\beta_1\colon\lambda_2\to\lambda_2??_{\rho}\lambda_3}$ and ${\beta_2\colon\lambda_3\to\lambda_2??_{\rho}\lambda_3}$ are the pushout-induced morphisms by ${\lambda_2+_{\lambda_0+\lambda_1}\lambda_3}$, we have two possible scenarios:
\begin{enumerate}
\item There is some unit $u$ in ${\lambda_2??_{\rho}\lambda_3}$ where ${\beta_1(u_2)=u=\beta_2(u_3)}$. So, there must also be a unit in ${\lambda_0+\lambda_1}$ identified with $u$. Thus, contradicting that ${\lambda_0+\lambda_1}$ has no units because ${U_0\sqcup U_1=\emptyset\sqcup\emptyset=\emptyset}$ by the fact that $\lambda_0$ and $\lambda_1$ are trivial computons (see Definitions \ref{def:markers} and \ref{def:coproduct}).
\item There is no unit in $\lambda_3$ identified with $\beta_1(u_2)$ through $\beta_2$. In this case, ${\lambda_2+_{\lambda_0+\lambda_1}\lambda_3}$ has ${|U_2|+|U_3|}$ units because ${|U_0|+|U_1|=0}$. As ${|U_2|\geq 1}$ because $\lambda_2$ is connected, ${|U_2|+|U_3|>|U_3|}$ so that ${\lambda_2+_{\lambda_0+\lambda_1}\lambda_3\cong\lambda_2??_{\rho}\lambda_3\cong\lambda_3}$ cannot hold.
\end{enumerate}
Proving that there is no right-identity is completely symmetric. Hence, the theorem holds.
\end{proof}

\section{Operational Semantics}
\label{sec:operation}

In this section, we describe execution semantics for arbitrary computons over a fixed set $\underline{n}$ of types and a fixed set $B$ of devices. For the sake of simplicity, we assume that a type is just a set of values, as in the formalisation of the operational semantics of coloured Petri nets \cite{jensencoloured1996}. As we are dealing with natural numbers to represent types, we require a deterministic way of mapping numbers to concrete types. This is done through the typing function described in Definition \ref{def:typing-function}.

\begin{definition}[Typing Function]\label{def:typing-function}
The typing function $T$ for a computon $\lambda$ has the signature ${\underline{n}\to\mathcal{U}}$ where ${\mathcal{U}}$ is the universe of types of some fixed type system. We assume ${\mathcal{C}}$ is in $\mathcal{U}$, which is the control type containing a single value $*$ that represents a control signal. 
\end{definition}

A typing function specifies the type of values a port can buffer. Definition \ref{def:computon-state} establishes that the collection of values associated to each port at time $j$ gives the state of a computon at $j$. 

\begin{definition}[Computon State] \label{def:computon-state}
The state of a computon $\lambda$ at $j$ is a function ${\delta^j\colon P\to(\bigcup_{x\in\underline{n}}T(x))\cup\{\bot\}}$ where, for all ${p\in P}$, ${\delta^j(p)}$ is of type ${T(c(p))}$. We say ${\delta^j}$ is initial if ${j=0}$, ${\delta^j(p)\neq\bot}$ for every ${p\in P^+}$ and ${\delta^j(q)=\bot}$ for all ${q\notin P^+}$. Here, we use $\bot$ to express value absence.
\end{definition}

In each state, units can be enabled or idle. Definition \ref{def:status} states that a unit is enabled at $j$ if all the ports connected to it have values assigned by ${\delta^j}$; otherwise, it is idle. When all units are idle, a final state is reached (see Definition \ref{def:computon-state-final}). 

\begin{definition}[Computation Unit Status]\label{def:status}
A unit ${u \in U}$ of a computon $\lambda$ is enabled at time $j$ if ${\delta^j(p)\neq\bot}$ for all ${p\in\bullet u}$; otherwise, $u$ is idle under ${\delta^j}$. 
\end{definition}

\begin{definition}[Termination] \label{def:computon-state-final}
A state ${\delta^j}$ of a computon $\lambda$ is final if each unit ${u\in U}$ is idle under ${\delta^j}$.
\end{definition}

When a parallel or a branching composite are in an initial state, multiple units can be enabled simultaneously. Particularly, in parallel computons, all enabled units are triggered simultaneously. In the case of branching, only one unit is chosen for activation at a time. To deal with such non-determinism, units are partitioned according to the source ports they share. For example, if $u_1$ and $u_2$ are two units with ${\bullet u_1=\bullet u_2}$, then they belong to the same partition. To select a concrete unit from each partition, we simply invoke the axiom of choice which non-deterministically chooses a representative unit. The collection of representatives yields the set of units ready for evaluation (see Definition \ref{def:ready}).

\begin{definition}[Ready Computation Units] \label{def:ready}
Given a computon $\lambda$, let ${E^j}$ be the finite set of computation units enabled under ${\delta^j}$ and $\sim$ be the equivalence relation on ${E^j}$ given by ${u_1\sim u_2\iff\bullet u_1=\bullet u_2}$ for all ${u_1,u_2\in E^j}$. If $A$ is the partition induced by $\sim$ and $h$ is the random choice function on $A$, ${\{h(E)\mid E\in A\}}$ is the set ${R^j}$ of computation units that are ready to be evaluated under ${\delta^j}$. 
\end{definition}

After forming the set ${R^j}$ under a state ${\delta^j}$, the computing devices of each ${R^j}$-unit are invoked to yield a new state at time ${j+1}$ (see Definitions \ref{def:evaluation} and \ref{def:transition}).

\begin{definition}[Computation Unit Evaluation]\label{def:evaluation}
Given an outflow ${o\in O}$ of a computon $\lambda$ and a state ${\delta^j}$, the result ${\llbracket f(o) \rrbracket^j}$ of evaluating a computing device ${f(o)}$ under ${\delta^j}$ is $f(o)(\delta^j(p_1),\ldots,\delta^j(p_k))$ such that there is a bijection ${\phi\colon\{1,\ldots,k\}\to s(r^{-1}(o))}$ satisfying ${p_m=\phi(m)}$ for all ${m\in\{1,\ldots,k\}}$. We say that the value ${\llbracket f(o) \rrbracket^j}$ is well-typed if and only if it is an element of ${(T\circ c \circ t)(o)}$.
\end{definition}

\begin{definition}[State Transition]\label{def:transition}
Given the state ${\delta^j}$ of a computon $\lambda$ at time ${j\geq 0}$, ${\delta^{j+1}}$ is given as follows for all ${p\in P}$:
\[\delta^{j+1}(p) = 
    \begin{cases}
       \ast & (\exists u \in R^j)[p\in u\bullet\text{ and }T(c(p))=\mathcal{C}]\\
       \llbracket f(o) \rrbracket^j & (\exists u \in R^j)(\exists o \in O)[\sigma(o)=u\text{ and }t(o)=p]\\       
       \delta^j(p) & (\nexists u \in R^j)[p\in\bullet u\cup u\bullet]\\ 
       \bot & \text{otherwise}
    \end{cases}
\]
\end{definition} 

At time ${j+1}$, a new state is constructed from ${\delta^j}$ according to the four cases considered by Definition \ref{def:transition}. The first case serves to store a control signal in each control port connected from each computation unit in ${R^j}$. The second one assigns the result of each computing device from each ${R^j}$-unit. The third case serves the role of a memory to keep untouched the values of those ports attached to idle units under ${\delta^j}$. The last case simply assigns the value $\bot$ to the ports connected to each unit in ${R^j}$, indicating that those (input) values have been consumed by ready units. These four cases collectively define a state transition during a computon's execution, a process that continues until reaching a state in which all units are idle. In other words, termination occurs when there is a finite orbit of states from the computon's initial state to a computon's final state. Although termination is not guaranteed in general, it is expected for composites.

\section{Implementation}
\label{sec:implementation}

We implemented the semantics of the computon model in Idris 2 \cite{bradyidris2021}, a functional programming language that treats (dependent) types as first-class entities, which allowed an almost direct realisation of the key semantic constructs. The first step in this process was to select a suitable representation for finite sets and total functions. For this, we found $\code{Fin m}$ adequate, which is a type with exactly $m$ inhabitants corresponding to the natural numbers ${0,\ldots,m-1}$. Totality of $\code{Fin}$ functions is a built-in property enforced by the Idris compiler. After deciding such representations, we implemented a suite of Idris functions to perform basic set operations such as fiber and image. With this, we subsequently implemented functions to compute disjoint union, union and pushout over $\code{Fin}$ types. 

Our purpose is to support automated composition via control-driven composition operators, so we are interested in the computational angle of colimit constructions. Accordingly, both union and disjoint union are characterised as records with four fields each: the cardinality of two operands and two canonical injections from the operands to the colimit construction being built. Cardinalities play an important role in our implementation since they define actual finite sets.\footnote{This simplification comes from the fact that the isomorphism class of a finite set $A$ can be identified with $|A|$.} Accordingly, considering that $\code{Fin m}$ has exactly $m$ monotonically increasing inhabitants, the disjoint union and union of $\code{Fin x}$ with $\code{Fin y}$ is simply $\code{x+y}$ and $\code{maximum(x,y)}$, correspondingly. The pushout of $\code{Fin}$ functions is a record that contains the cardinality of the pushout object as well as the two pushout-induced functions into the pushout object. Automatically constructing such functions is done through a concrete algorithm.

The implementation of pushout and coproduct over $\code{Fin}$ forms the basis on which the notion of a category of computons is implemented upon. A computon object $\lambda$ is particularly a record with 15 fields, holding the cardinalities of $U$, $P$, $I$, $O$ and $\underline{n}$ as well as the total functions $s$, $t$, $\sigma$, $\tau$, $c$, $r$ and $f$. The cardinality of $B$ is not needed since we use the $\code{String}$ type to represent computing devices. Accordingly, ${f\colon O\to B}$ is an Idris function ${\code{Fin o}\to \code{String}}$, where $\code{o}$ is the cardinality of $O$. In addition to the components from Definition \ref{def:computon}, the \code{Computon} record includes fields to prove the non-emptiness of $P$ and $\underline{n}$. By the definition of a $\code{Fin n}$ type, having ${n>0}$ entails that the number $0$ is always included. Proofs of surjectivity for $\sigma$, $\tau$ and $r$ are not required within the \code{Computon} record since we only deal primitives and trivials explicitly from which composites are built. Thus, such proofs are just needed for the record representing primitive computons, which in addition requires proofs of injectivity for $s$ and $t$ and proofs that $|U|=1$ and ${P\cong I\sqcup O}$.\footnote{If $\code{Fin p}$, $\code{Fin i}$ and $\code{Fin o}$ are the types representing the finite sets $P$, $I$ and $O$ of a primitive computon $\lambda$, $P\cong I\sqcup O$ reduces to verifying $\code{p=i+o}$.} The record representing trivial computons does not require surjectivity proofs at all, since such a property holds directly from Definition \ref{def:trivial-computon}. In this case, it is sufficient to keep fields for the proofs of ${|U|=|I|=|O|=0}$.

Unlike computons, a computon morphism is specified as a record with two arguments for delimiting domain and codomain. Apart from the six components from Definition \ref{def:morphism}, a \code{Morphism} record requires a proof that the function between computing devices is an inclusion as well as a proof that both boundary conditions from Definition \ref{def:morphism} are met.\footnote{The proof of inclusion for the $\underline{n}$-component is constructed dynamically to reduce proof tasks for developers.} This record forms the basis of a computon monomorphism which also is a record but with proofs that all the components from Definition \ref{def:morphism} are injective. A monomorphism serves in turn as the basis for specifying in- and out-markers. The former requires a mono from a trivial computon into all the inports of the codomain computon. The latter is similar but requires a proof that the trivial computon can be inserted into all the outports of the codomain. 

To assist developers in forming trivials, primitives, (mono)morphisms and markers, we provide Idris functions that facilitate the construction of such objects. To specify colimits, we also define records for spans, coproducts and pushouts. Coproduct and pushout are simply computed as in Definitions \ref{def:coproduct} and \ref{def:pushout} through a concrete algorithm. By offering functions to compute coproducts and pushouts, we provide a basis for the implementation of the elementary composition operators described in Sect. \ref{sec:operators}, namely total/partial sequencing, asynchronous parallelising and closed/open branching. The operator for synchronous parallelising is not implemented since this is built out of partial sequencing and asynchronous parallelising. The source code of the implemented programming environment, together with examples, are publicly available at \url{https://github.com/damianarellanes/computons-idris}. 

As the environment was implemented in Idris, the universe of data types for the operational semantics corresponds to it. To avoid parsing the static structure of a computon at run-time, we defined a data type that embodies a structural simplification, using vectors for direct access. For a computon $\lambda$, such a simplification stores $P^-$, a vector that specifies which ports buffer control, vectors for mapping each ${u \in U}$ to ${\bullet u}$ and to ${u\bullet}$, a vector from each ${p\in t(O)}$ to the unit-indexed computing devices that $p$ reads from and a vector that specifies which units are enabled at a particular point in time. A state is simply a vector of size ${|P|}$ where each index represents the value of each port. Such values are of \code{IO} type because they have side effects, as a result of treating each computing device as a string that represents the network endpoint of a behaviour given in the form of a web service. That is, computing devices are implemented as web services. We treat devices in this way to enable interoperability while avoiding colimit computations over large strings or source code compilation.\footnote{By supporting interoperability through web services, one computing device can be written in Java, another in Haskell and another in Python, just to give a few examples. That is, our implementation enables a practical hybrid model of computation.} Instead, $B$-elements are simply IP addresses contacted upon evaluation as per Definition \ref{def:evaluation}. Evaluating a computation unit yields a new vector/state to be processed in the next time step, as prescribed by Definitions \ref{def:ready} and \ref{def:transition}. This process continues until reaching a final state of the computon being executed.

\section{Related Work}
\label{sec:related-work}

The computon model, originally introduced in \cite{arellanescompositional2026}, describes composition operators for total/partial sequencing, synchronous/asynchronous parallelising, closed branching and head- and tail-iteration. Unfortunately, such an original formulation is a generic reference rather a concrete realisation that dictates how to compute at the low-level. Accordingly, the original computon definition neither defines functional relations between inflows and outflows nor imposes particular flow order within computation units, leading to a situation in which a unit $u$ has to either replicate the same value to all the data ports in ${u\bullet}$ or enforce ${u\bullet}$ to have only one data port (for holding a single datum or a product type value). To remediate this structural inflexibility with operational consequences, Definition \ref{def:computon} equips computons with a set $B$ of computing devices, together with a function chain ${I\overset{\text{r}}{\twoheadrightarrow}O\xrightarrow{f}B}$ that satisfies ${\sigma\circ r=\tau}$ to encapsulate flows within units. 

Unlike \cite{arellanescompositional2026}, Sect. \ref{sec:operation} describes operational semantics to dictate how devices are triggered within units. In \cite{arellanescompositional2026}, there is no concrete implementation or programming environment, and execution is described generically at a higher-level of abstraction via P/T Petri nets. Although similar operational semantics to ours have been proposed in \cite{arellanesboolean2026}, such a work handles units homogeneously in the form of families of NAND operators to show that the computon model is able to perform any Boolean function when treated as a non-uniform model of computation. In our work, a computation unit is a collection of computing devices each represented as a finite sequence of symbols over a finite alphabet.

A crucial difference with respect to \cite{arellanesboolean2026} and \cite{arellanescompositional2026} is that they do not provide any operator for open branching, needed to enhance flexibility towards more expressive decision-making composites. In the present work, we found that there is no need to provide an explicit operator for synchronous parallelising since such a behaviour can be constructed compositionally out of partial sequencing and asynchronous parallelising (see Sect. \ref{sec:operators-parallel}). Although we do not offer operators for looping due to the lack of space, they can be easily integrated and implemented within our framework. Unlike \cite{arellanescompositional2026}, our work studies identity laws for composition operators and replaces the original formulation of sequencing to enable this property. In particular, connectivity in the sense of Definition \ref{def:connected} is not further required.

Beyond the original computon model, there have been other MHCs that treat control flow explicitly. Process algebras \cite{baetenprocess2009,middelburgimperative2024} form perhaps the most prominent family of MHCs which provide concrete operators to compose processes by control flow. Although data flow can separately be observable and analysable, there is no explicit support for partial sequencing. String diagrams \cite{piedeleuintroduction2025} are becoming increasingly popular to specify high-level processes formally since they are grounded in categorical semantics for reasoning about computation in a compositional manner. As syntax, they provide a graphical notation for monoidal categories by visually encoding composition rules: wires denote data flows coming into/from boxes which, in turn, represent processes. Although there are no special wires for representing control flow, string diagrams support sequential, asynchronous parallel and iterative composition through morphism composition, tensor product and traces \cite{selingersurvey2011}. Sequencing is only total and there is no support for synchronous parallelising. Recently, (probabilistic) branching has been introduced \cite{villoriaenriching2025} and there have been efforts to colouring wires, but not for separating concerns \cite{earnshawstring2023}. In any case, string diagrams assume that data follows control, as evidenced by \cite{bonchidiagrammatic2025}; consequently, they do not support situations in which data reaches computing devices independently from control signals.

Workflow nets \cite{vanderaalstsoundness2011} allow the formal specification of control-driven computing devices, but do not offer any separation of concerns and, unlike process algebras and string diagrams, they do not provide any formal operators to compositionally form control-driven composites. Exogenous Connectors \cite{lauexogenous2005} and Behaviour Trees \cite{colledanchisebehavior2018} alleviate this composition issue by providing separate operators for sequencing, branching, parallelising and looping but, unfortunately, the data dimension is left implicit. In the case of \cite{lauexogenous2005}, there have been attempts to separate data from control, albeit without any formal semantics \cite{arellanesdecentralized2023,laucomponent2011}. Prosave \cite{vulgarakisformal2009} and SCADE \cite{colacoscade2017} offer such a separation, but also in an informal manner.

\section{Conclusions and Future Work}
\label{sec:conclusions}
\vspace{-5pt}
In this paper, we introduced and extended a novel MHC, referred to as the computon model, in which (trivial and primitive) computons are the fundamental building blocks. A trivial computon is intuitively a collection of typed ports, whereas a primitive one has a unique computation unit that encapsulates a collection of (potentially interrelated) computing devices, only accessible through a well-defined port-based interface. Computon interfaces have ports to buffer control signals, whereas data ports are optional. This deliberate design facilitates inductive composition for forming explicit control flow structures for the invocation of computing devices in some order. In Sect. \ref{sec:operators}, we described finite colimit constructions for composition operators to allow the formation of total/partial sequential computons, synchronous/asynchronous parallel computons and open/closed branching ones. In that section, we showed that total sequencing satisfies associativity and identity, whereas partial sequencing only fulfils identity. Asynchronous parallelising and closed branching are both associative and commutative. Only open branching satisfies all the three properties. In the future, we plan to study further algebraic laws (e.g., inversibility) and provide extra colimit constructions to support other forms of composition such as conditional looping. Offering more complex forms of (compositional) concurrency is also a future direction.

In any theoretical extension, control flow must still be a first-class composition dimension because it is what gives rise to the notion of computation. Although control is always present in any low- or even high-level computing device \cite{arellanesmodels2025}, data flow must not be neglected but it must be treated as a secondary dimension always governed by control \cite{tripakismodular2013}. Governance does not imply that data follows control as, in some high-level devices, control signals may arrive before data values, leading to constituent computing devices waiting until receiving all the input data they need. This situation can particularly arise in partial sequencing or asynchronous parallelising. Despite that data does not necessarily follow control, computation order is still predictable in our proposal due to the synchronous evaluation of computation units, as described in Sect. \ref{sec:operation}. More precisely, Sect. \ref{sec:operation} stipulates that a computation unit remains idle until receiving information in all the ports connected to it. When all input information becomes available, an idle unit transitions to an enabled state. Given that non-determinism is introduced by branching composites, multiple units can be enabled simultaneously in which case only one is chosen for evaluation arbitrarily. In the future, we would like to extend the computon model with support for probabilistic choice and enhance its expressivity through conditional evaluation. Enabled units are currently evaluated by triggering their encapsulated computing devices on the inputs they receive from the external environment or from other computons. Upon evaluation, new values are stored in the ports connected from those units, in order to reach a new state. This process continues until reaching a state in which all units become idle.

The operational semantics described in Sect. \ref{sec:operation} can be interpreted in a categorical setting for a more rigorous study of computon execution. Operational semantics is not the focus of this paper but colimit-based composition is. In the future, we plan to formalise operational semantics in the language of symmetric monoidal categories by defining ports and computing devices as objects and morphisms of a category, respectively. Although we use categorical semantics to formalise our proposed model, it is important to mention that other formal languages can be used instead. That is, the proposed MHC is independent of its specification, just as it occurs with any other foundational model.

To bridge theory and practice, we implemented the categorical semantics of the computon model in Idris 2, including the formalisation of computons, computon morphisms and colimit constructions. The implementation yields a programming environment for forming control-driven composite computons with strong structural correctness by construction. We envision that primitives and composites could populate a repository of computational behaviours to be reused across several domains. For example, in a real e-commerce scenario, there could be primitive computons for payment processing and email delivery which can altogether be reused in both travel and healthcare applications. The computon model is general enough to be used not only in software engineering but in any domain requiring computable objects, e.g., it has already been applied in Artificial Intelligence for compositionally forming a Long Short-Term Memory \cite{arellanescompositional2026}, in biomedicine for modelling the oscillatory dynamics of a mdm2-p53 regulatory pathway \cite{arellanesboolean2026} and in electronic engineering for compositional circuit design \cite{arellanesboolean2026}. Given the generality of the computon model, several other application domains could benefit from it, e.g., quantum mechanics and linguistics for compositionally describing quantum teleportation and sentence meaning, respectively. 

The implemented programming environment currently provides a basis to support cross-domain applications. In our short-term vision, end-users just have to select computons from a global repository to compose more complex computons via the control-driven composition operators presented in Sect. \ref{sec:operators}. Although our proposal enables computons to be semantically correct by construction, this is not sufficient to achieve correct reusability. For this, we need guarantees of behaviour correctness with respect to some specification, so one can equip computons with correctness proofs before storing them in a global repository. Full computon certification is an aspect we plan to investigate in the future. We particularly believe that the bottom-up composition approach enforced by the computon model could facilitate \emph{a priori} certification \cite{costantiniensuring2022}, i.e., if a computon is guaranteed to meet its specification, it will remain correct even if it ever becomes part of a composite. For example, if computons $\lambda_1$ and $\lambda_2$ are certified, the specification of $\lambda_1\unrhd_\rho\lambda_2$ will be predictable prior composition, leading to predictable system assembly.

\bibliographystyle{eptcs}
\bibliography{refs}

\end{document}